\pdfoutput=1
\documentclass{article}

    \PassOptionsToPackage{numbers, compress}{natbib}
 \usepackage[preprint]{neurips_2026}

\usepackage[utf8]{inputenc} 
\usepackage[T1]{fontenc}    
\usepackage{hyperref}       
\usepackage{url}            
\usepackage{booktabs}       
\usepackage{amsfonts}       
\usepackage{nicefrac}       
\usepackage{microtype}      
\usepackage{xcolor}         

\usepackage{graphicx}
\usepackage{subcaption}
\usepackage{multirow}
\usepackage{array}
\newcolumntype{C}{>{\centering\arraybackslash}m{3.1em}}
\usepackage{colortbl}
\definecolor{morseblue}{HTML}{DCE8F4} 
\usepackage{amsmath}
\usepackage{amssymb}
\usepackage{mathtools}
\usepackage{amsthm}
\usepackage{algorithm}
\usepackage{algorithmic}
\usepackage{float}
\usepackage[capitalize,noabbrev]{cleveref}
\usepackage{enumitem}
\usepackage{wrapfig}

\newcommand{\ensureroom}[1]{\par\begingroup
  \ifdim\pagegoal=\maxdimen\else
    \ifdim\dimexpr\pagegoal-\pagetotal\relax<#1\relax\newpage\fi
  \fi\endgroup}

\theoremstyle{plain}
\newtheorem{theorem}{Theorem}[section]
\newtheorem{proposition}[theorem]{Proposition}

\newtheorem{corollary}[theorem]{Corollary}
\theoremstyle{definition}

\newtheorem{assumption}[theorem]{Assumption}
\theoremstyle{remark}

\title{MoRSE: Task-Oriented Multi-Agent System\\with Mixture of Role-Subtask Experts}
\author{%
  \textbf{Peiwen Li} \quad \textbf{Shiyang Zhang} \quad \textbf{Yangtian Zhang} \quad \textbf{Sizhuang He} \\
  \textbf{David van Dijk} \quad \textbf{Rex Ying} \\
  Yale University, New Haven, CT, USA \\
  \texttt{\{peiwen.li, shiyang.zhang, yangtian.zhang, sizhuang.he\}@yale.edu} \\
  \texttt{\{david.vandijk, rex.ying\}@yale.edu}
}

\begin{document}

\maketitle

\begin{abstract}
  Large language model-based multi-agent systems have recently shown strong potential for complex, long-horizon tasks.
However, existing methods mainly rely on coarse prompt-level differentiation without parameter adaptation for diverse subtasks, resulting in insufficient inter-agent heterogeneity and limited specialized capability that bottleneck performance on tasks with complex requirements.
To address this, we introduce a Task-Oriented \textbf{\underline{M}}ulti-Agent System with Mixture \textbf{\underline{o}}f \textbf{\underline{R}}ole-\textbf{\underline{S}}ubtask \textbf{\underline{E}}xperts (\textbf{MoRSE}) that distinguishes agents with \emph{(role, subtask)-conditional} specialization at both the task structure and parameter levels.
To make agents' responsibility explicit at the task structure level, we formulate a task-oriented multi-agent system that decomposes each task into a dependency-aware Directed Acyclic Graph of subtasks and assigns each agent a specific (role, subtask), introducing task-level specialization across collaborating agents.
 Additionally, to address the diverse role and subtask parameter adaptation demands, we propose a dynamic Mixture of (role, subtask) LoRA Experts module with a prototype-based semantic router for subtasks, augmenting agents with parameter-level specialization on a shared LLM substrate cost-effectively.
  Then, to co-optimize experts and router stably under sparse task rewards, we further propose a hierarchical group-relative policy optimization with two-layer credit assignment that isolates expert updates from the cross-route variance introduced by routing decisions, disentangling expert quality from routing quality.
  Experiments on code-generation benchmarks across three backbones demonstrate the effectiveness of our approach, with improvements in both whole-task and step-wise performance, and the gains from trained specialization generalize across held-out task categories and domains.

\end{abstract}

\section{Introduction}
\label{sec:introduction}

Large language model (LLM)-based multi-agent systems (MAS) have emerged as an effective paradigm for addressing complex, long-horizon tasks, by orchestrating multiple role-specialized agents into structured multi-stage pipelines~\citep{wu2024autogen,hong2024metagpt,qian2024chatdev}. This paradigm has demonstrated strong performance across diverse domains such as collaborative software development with simulated engineering teams~\citep{hong2024metagpt,qian2024chatdev} and multi-step reasoning through inter-agent debate and critique~\citep{du2023debate}, with recent work further pushing collaboration to populations of hundreds or even thousands of agents~\citep{qian2024macnet}. Yet simply adding more agents to a shared backbone yields diminishing returns~\citep{yang2026diversity}; the real bottleneck is \emph{how distinctly each agent contributes to the joint solution}.

\looseness=-1
Existing MAS methods, however, pursue such distinctness mainly through prompt-level
role descriptions over frozen base models~\citep{wu2024autogen,hong2024metagpt,qian2024chatdev,qian2024macnet};
but without explicit subtask assignment, agents tend to produce \emph{redundant},
overlapping responsibilities and correlated outputs that saturate performance
quickly as the system scales~\citep{yang2026diversity}, leaving inter-agent
\emph{heterogeneity insufficient at the task structure level}. 
A growing line of work
moves beyond prompts toward parameter-level adaptation, including role-conditioned
adaptation in a single-agent setting~\citep{moragent}, and multi-agent reinforcement
learning (RL) frameworks like AT-GRPO~\citep{stronger_mas} and MARTI~\citep{marti}
that extend RL optimization to multi-agent settings. Yet these methods adapt mainly at the predefined and limited role level; the diverse subtasks within a single complex task impose distinct demands that a single shared base model cannot satisfy, leaving inter-agent \emph{heterogeneity insufficient at the parameter level} and bottlenecking performance on tasks with complex requirements.

To address these two levels of insufficiency, in this paper we study \emph{sufficient agent
heterogeneity} on complex open-ended multi-agent collaboration tasks from both
the \emph{task structure} and \emph{parameter} levels. We introduce a
Task-Oriented \textbf{\underline{M}}ulti-Agent System with Mixture \textbf{\underline{o}}f \textbf{\underline{R}}ole-\textbf{\underline{S}}ubtask \textbf{\underline{E}}xperts
(\textbf{MoRSE}) that explicitly manages subtasks within the whole task and
learns role- and subtask-specific parameter specialization across collaborating
agents cost-effectively on a shared backbone via a stable optimization scheme.
Here, an agent's \emph{role} denotes its reusable function drawn from a small
discrete set (e.g., merging upstream context or executing a step), whereas its
\emph{subtask} denotes the instance-specific unit of work described in natural
language; the two axes impose complementary specialization demands and jointly
define each agent's responsibility.

\looseness=-1
Specifically, we introduce a task-oriented MAS framework that exposes each agent's
(role, subtask) responsibility, reducing task-structure-level redundancy
and providing a substrate for fine-grained parameter-level specialization. Inspired by
graph-based planning~\citep{wu2025gap}, our model decomposes each task instance into a dependency-aware Directed Acyclic Graph (DAG) of subtasks and assigns each agent a specific (role, subtask). 
Building on this framework, we further propose a dynamic mixture of (role, subtask) LoRA experts that enhances inter-agent heterogeneity at the parameter level.
It attaches role and subtask experts selected per agent call by a learned prototype-based semantic router.
Jointly optimizing the experts and router under sparse, open-ended task rewards, however, can entangle expert quality with routing quality.
Finally, we propose a hierarchical group-relative policy optimization with
two-layer credit assignment that splits the task-level reward into a
within-route advantage updating experts within a routed combination
and a cross-route advantage updating the router across combinations.
We show analytically that it reduces gradient variance for expert updates
and empirically that it stabilizes co-optimization, reliably realizing the
parameter-level heterogeneity gains.

Experiments on complex code-generation benchmarks across multiple LLM
backbones demonstrate that \textbf{MoRSE} achieves the best held-out test
scores over all baselines in both whole-task and step-wise performance.
On held-out task categories, \textbf{MoRSE} further surpasses both its
untrained variant and a standard fixed-LoRA fine-tuning baseline, indicating
that the learned role-subtask experts generalize better under distribution
shift.
Our contributions are summarized as follows:
\begin{itemize}[leftmargin=1.5em, itemindent=0pt]
    \item We 
    introduce a
    task-oriented MAS framework with dependency-aware DAG decomposition and
    per-agent (role, subtask) assignment as substrate to reduce redundancy at the task-structure level. 

    \item We propose a dynamic mixture of (role, subtask)-conditioned LoRA
    experts on a shared backbone with a prototype-based semantic router for
    effective parameter-level specialization.

    \item We propose a hierarchical group-relative policy optimization with
    two-layer credit assignment that disentangles expert quality from routing
    quality to stabilize co-optimization.\footnote{Code: \url{https://github.com/lpwpower/MoRSE}.
    }
\end{itemize}

\section{Preliminaries}
\label{sec:preliminary}

\subsection{Empirical Motivation}
\label{subsec:empirical-motivation}

\looseness=-1
To diagnose the heterogeneity bottleneck identified in Sec.~\ref{sec:introduction} and motivate the design of \textbf{MoRSE}, we conduct two preliminary studies on SRDD~\citep{qian2024chatdev} with Qwen3-4B. The first measures the redundancy of role-prompted MAS to confirm the heterogeneity deficit is real, and to ground our adoption of subtask-based decomposition as the framework substrate. The second probes the limitation of parameter-efficient fine-tuning on a single base LM in satisfying the heterogeneous learning demands of diverse roles and subtasks, directly motivating the design of Mixture of Role-Subtask LoRA Experts.

\paragraph{Redundancy in role-prompted, graph-structured MAS.}
\label{subsec:prelim-redundancy}

For each SRDD task instance involving $n$ collaborating agents, we measure heterogeneity by the mean pairwise cosine similarity of node-level prompts $\{p_i\}_{i=1}^n$ and outputs $\{y_i\}_{i=1}^n$ under a frozen text encoder $f(\cdot)$ (gte-Qwen2-7B-instruct):
\begin{equation}
\mathrm{Sim}(\{t_i\}) \triangleq \frac{2}{n(n-1)} \sum_{i<j} \cos\!\big(\hat{f}(t_i),\hat{f}(t_j)\big),
\qquad t_i \in \{p_i\}_{i=1}^n \;\text{or}\; \{y_i\}_{i=1}^n,
\label{eq:sim-def}
\end{equation}
producing $\mathrm{Sim}_{\text{prompt}}$ and $\mathrm{Sim}_{\text{output}}$, where smaller values indicate lower redundancy. Figure~\ref{fig:prelim-p1-srdd-qwen} traces, for every SRDD instance, the per-sample shift from $\mathrm{Sim}_{\text{prompt}}$ (top axis row) to $\mathrm{Sim}_{\text{output}}$ (bottom axis row) under three methods: MacNet~\citep{qian2024macnet} without and with its review stage, both of which differentiate role-based agents at the \emph{prompt} level over \emph{static, predefined} graph topologies; and our \emph{task-decomposed dynamic} graph (\textbf{MoRSE}$_{\text{base}}$, the framework substrate of MoRSE without trained MoLE experts) that formulates the DAG at the $(\text{role}, \text{subtask})$ level. Both MacNet variants \emph{collapse} the per-sample distribution toward $\mathrm{Sim}_{\text{output}} \approx 1$ (median shift $\Delta_{\text{out-prompt}}{=}+0.14$ for both, indicating that the review stage does not relieve the redundancy). In contrast, MoRSE$_{\text{base}}$ \emph{spreads} the per-sample outputs apart ($\Delta{=}-0.15$, lowering median $\mathrm{Sim}_{\text{output}}$ from $0.75$ to $0.60$). However, even with our subtask decomposition, the median $\mathrm{Sim}_{\text{output}}$ still sits around $0.60$, confirming that structural decomposition alone does not close the heterogeneity deficit---motivating the parameter-level adaptation we develop in Sec.~\ref{sec:method}, which the next study further grounds in the base LM.

\begin{figure}[!htb]
\centering
\includegraphics[width=\linewidth]{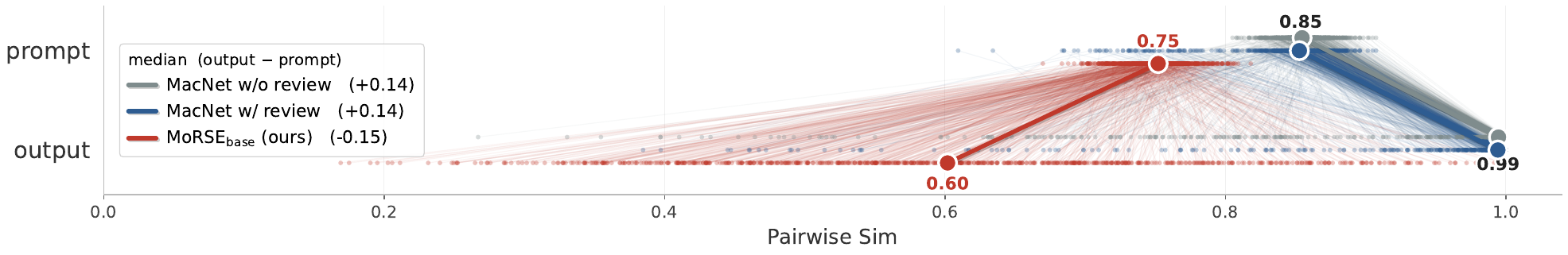}
\caption{\textbf{Heterogeneity deficit in role-prompted, graph-structured MAS.} Each thin line is one SRDD instance, drawn from its prompt-level $\mathrm{Sim}$ (top row) to its output-level $\mathrm{Sim}$ (bottom row); bold lines show per-method median trajectories. Smaller $\mathrm{Sim}$ indicates lower redundancy. The legend reports the median shift $\Delta=\mathrm{Sim}_{\text{output}}-\mathrm{Sim}_{\text{prompt}}$: positive (+) means the pipeline \emph{amplifies} agent-output redundancy, negative (-) means it \emph{disperses} agents apart.}
\label{fig:prelim-p1-srdd-qwen}
\end{figure}

\paragraph{Role-subtask mismatch in the base LM.}
\label{subsec:prelim-mismatch}

\looseness=-1
To inform parameter-efficient fine-tuning over the dynamic $(\text{role}, \text{subtask})$-DAG, we ask whether different roles and subtasks require mismatched rank-$\rho$ adaptations of the base LM that no single shared LoRA can simultaneously provide under its rank-$\rho$ bottleneck, a limitation we formalize in Theorem~\ref{thm:lora-subspace-conflict}, Appendix~\ref{app:role-subtask-mismatch}. Adapting subspace-based representation analyses~\citep{raghu2017svcca,kornblith2019similarity,hamm2008grassmann} to per-$(\text{role}, \text{subtask})$ groups, we take \texttt{execute} and \texttt{merge} as two example roles: Figure~\ref{fig:prelim-role-subtask-mismatch}(a) shows that real role labels yield substantially larger pairwise rank-$\rho$ subspace divergence ($1-\|V_a^{\top}V_b\|_F^2/\rho$, $\rho{=}8$) than a label-shuffled binary baseline across the deeper 8 transformer layers (where prior work shows task-relevant representations are most concentrated~\citep{tenney2019bert,meng2022locating}; gap $\Delta{\approx}0.40$). Fixing the role to \texttt{execute} and varying subtasks, we further cluster SRDD subtask descriptions into $K{=}8$ groups and apply the same diagnostic across the same deeper layers: real cluster labels yield divergence ${\sim}0.75$ vs.\ ${\sim}0.45$ under label-shuffled baselines, with $\sigma$-bands non-overlapping at every layer (Figure~\ref{fig:prelim-role-subtask-mismatch}(b)) and visually stable cluster-pair structure across layers (Figure~\ref{fig:prelim-role-subtask-mismatch}(c) for one layer; full per-layer panel in Appendix~\ref{app:role-subtask-mismatch}). Full procedure for both diagnostics is given in Appendix~\ref{app:role-subtask-subspace-procedure}. Together, these results reveal specialized adaptation needs along two axes---\emph{roles} (persistent, process-level) and \emph{subtasks} (fine-grained, semantic)---directly motivating a dynamic mixture of LoRA experts with a router for $(\text{role}, \text{subtask})$ specialization on a shared base model (Sec.~\ref{subsec:mole}).

\begin{figure}[!htb]
\centering
\begin{subfigure}{0.36\linewidth}
\centering
\includegraphics[width=\linewidth]{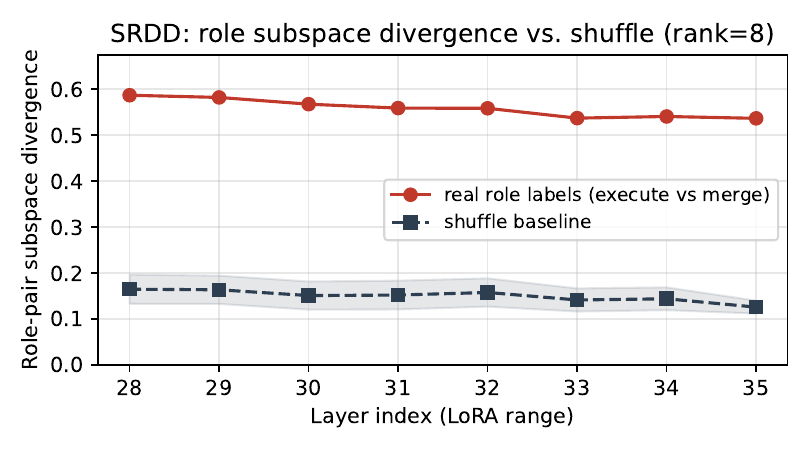}
\caption{Role subspace divergence.}
\end{subfigure}
\hfill
\begin{subfigure}{0.36\linewidth}
\centering
\includegraphics[width=\linewidth]{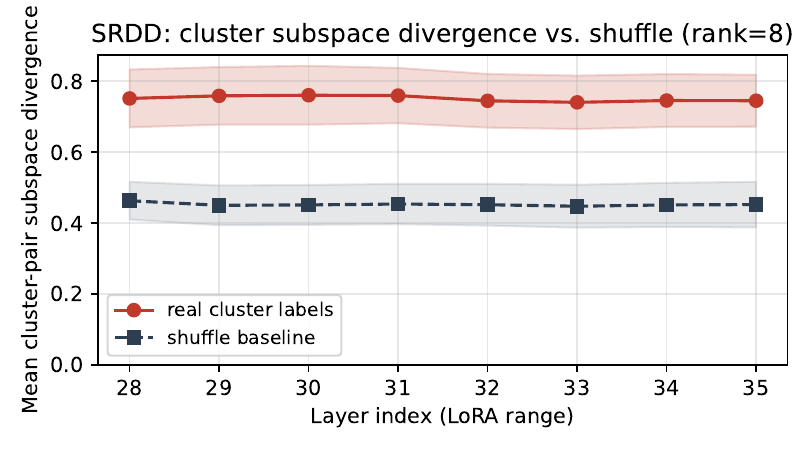}
\caption{Subtask subspace divergence.}
\end{subfigure}
\hfill
\begin{subfigure}{0.2\linewidth}
\centering
\includegraphics[width=\linewidth]{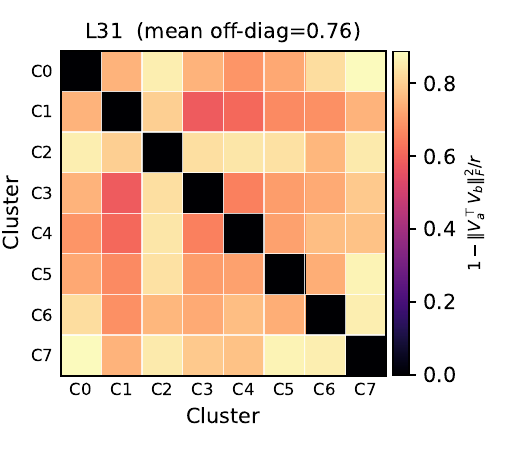}
\caption{Cluster-pair structure example.}
\end{subfigure}
\caption{\textbf{Role and subtask mismatch in the base LM} (Qwen3-4B, SRDD). (a) Pairwise rank-$8$ subspace divergence between two roles (\texttt{execute} and \texttt{merge}) per LoRA-injected layer; the dashed line shows the label-shuffled binary baseline. Higher = more divergent subspaces. (b) Same diagnostic with $K{=}8$ subtask clusters; shaded bands show label-shuffled baselines. (c) Heatmap of pairwise cluster divergence at layer~31 (representative); brighter = more divergent.}
\label{fig:prelim-role-subtask-mismatch}
\end{figure}

The two studies together reveal further heterogeneity demand at the \emph{task-structural} (output redundancy) and \emph{parameter} (role/subtask mismatch) levels, which we resolve with a dynamic $(\text{role}, \text{subtask})$-DAG and a router-composed mixture of LoRA experts on a shared base model (Sec.~\ref{sec:method}).

\subsection{Problem Formulation}
\label{sec:problem}

In this paper, we consider a complex task instance $\mathcal{T}$ that admits decomposition into $|V|$ interdependent subtasks $\{v_i\}_{i=1}^{|V|}$ with natural-language descriptions $\{s_i\}_{i=1}^{|V|}$. Each subtask $v_i$ is handled by a collaborating agent $a_i$ with role $r_i$ drawn from a small discrete role set $\mathcal{R}$. Agent $a_i$ produces an artifact $y_i$ conditioned on its $(\text{role}, \text{subtask})$ pair $(r_i, s_i)$ and the upstream artifacts $\{y_j\}_{j \in \mathrm{Pa}(i)}$, where $\mathrm{Pa}(i) \subseteq \{1, \ldots, |V|\}$ indexes the subtasks that $v_i$ depends on.

All agents share a frozen base language model with parameters $\theta_0$, and $(\text{role}, \text{subtask})$-conditional heterogeneity is realised via parameter-efficient adaptation $\theta_i = \theta_0 + \Delta\theta_i(r_i, s_i)$ with $\dim \Delta\theta_i \ll \dim \theta_0$, so that $y_i \sim p(\cdot \mid r_i, s_i, \{y_j\}_{j \in \mathrm{Pa}(i)}; \theta_i)$. Given a reward $R(\{y_i\}_{i=1}^{|V|})$ defined over the joint outputs, we jointly optimise all learnable parameters $\Theta$ of the chosen $\Delta\theta$ parameterisation to maximise expected reward over a task distribution $\mathcal{D}$:
\begin{equation}
\label{eq:objective}
    \max_{\Theta} \; \mathbb{E}_{\mathcal{T} \sim \mathcal{D},\;
    \{y_i\} \sim p_{\Theta}(\cdot \mid \mathcal{T})} \big[
    R(\{y_i\}) \big].
\end{equation}

\section{Methodology}
\label{sec:method}

\begin{figure*}[!b]
\centering
\includegraphics[width=\linewidth]{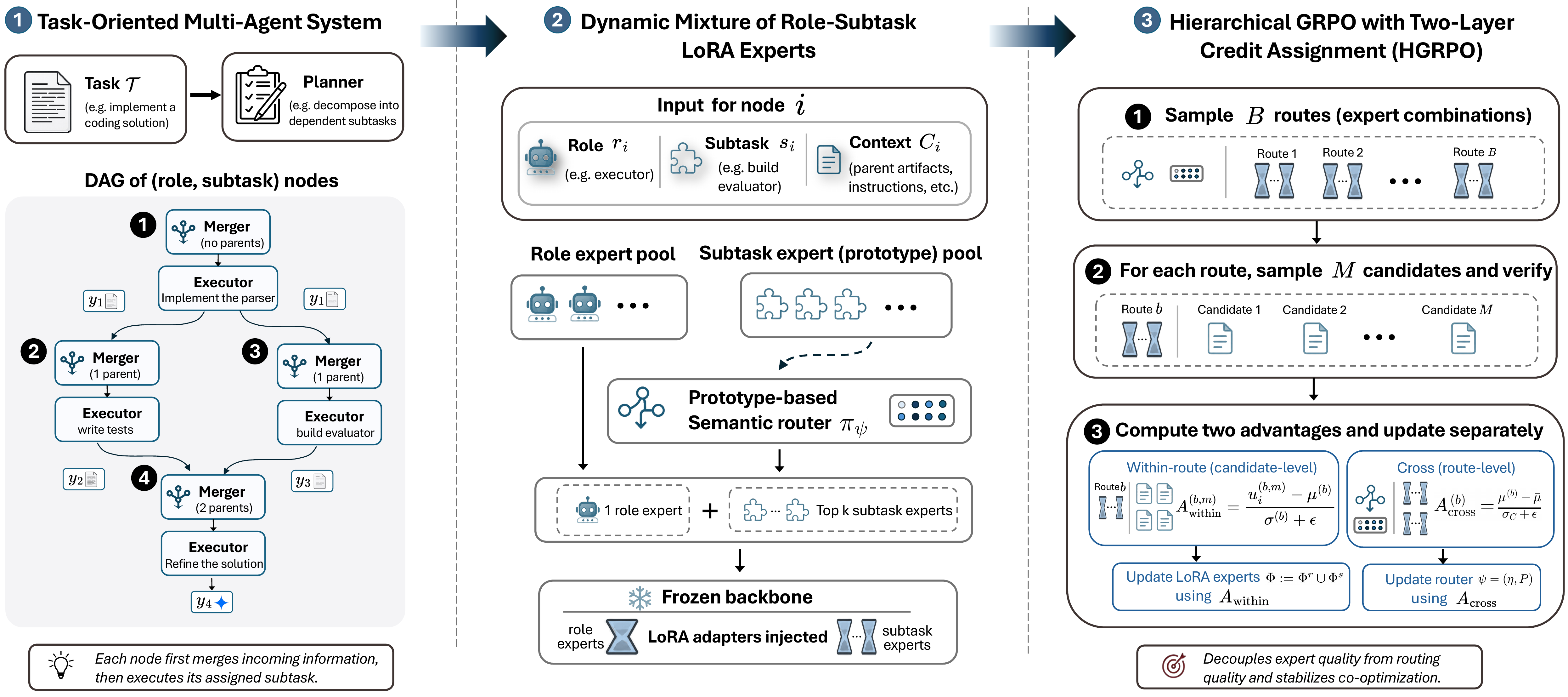}
\caption{Overview of the MoRSE framework: (1) task-oriented MAS that decomposes $\mathcal{T}$ into a per-instance $(\text{role}, \text{subtask})$-DAG with per-node Merger$\to$Executor execution; (2) dynamic mixture of role-subtask LoRA experts (role expert + top-$K$ subtask experts) on a shared frozen backbone; (3) hierarchical GRPO with two-layer credit ($A_{\text{within}}, A_{\text{cross}}$) that decouples expert and router updates.}
\label{fig:tomas-framework}
\end{figure*}

\looseness=-1
We instantiate the Problem Formulation in Sec.~\ref{sec:problem} with three tightly coupled components realising $(\text{role}, \text{subtask})$-conditional specialization at both the task-structure and parameter levels (Figure~\ref{fig:tomas-framework}). We first introduce a task-oriented multi-agent system that decomposes each task into a subtask DAG with a rule-based step verifier, supplying the per-node $(r_i, s_i)$ labels and step rewards required by the downstream modules (Sec.~\ref{subsec:tomas}). On the shared frozen backbone, we then build a dual-factorised mixture of role and subtask LoRA experts composed by a prototype-based semantic router, specialising at the parameter level across an \emph{open-ended subtask space} at near single-model cost (Sec.~\ref{subsec:mole}). Finally, we propose hierarchical GRPO with two-layer credit assignment that decouples expert and router updates and provably reduces the LoRA-expert gradient variance for stable co-optimization (Sec.~\ref{subsec:hgrpo}).

\subsection{Task-Oriented Multi-Agent System (ToMAS)}
\label{subsec:tomas}
To address the task-structure-level redundancy diagnosed in Sec.~\ref{subsec:prelim-redundancy}, we explicitly decompose each task instance into subtasks structured by dependencies, inspired by graph-based planning~\citep{wu2025gap}, and execute the resulting DAG with a step-level rule-based verifier at every node. This produces two signals required by the downstream modules: per-node $(r_i, s_i)$ labels that condition parameter specialization (Sec.~\ref{subsec:mole}), and step-level rewards $\{u_i\}_{i \in V}$ that supply the training signal for that specialization via per-step credit assignment (Sec.~\ref{subsec:hgrpo}).

\paragraph{DAG-based task decomposition and execution.}
Given a task instance $\mathcal{T}$, a planner dynamically produces a per-instance dependency-aware DAG $\mathcal{G} = (V, E)$, where each node $v_i \in V$ carries a subtask description $s_i$ and each directed edge $(v_i, v_j) \in E$ indicates that subtask $j$ depends on the artifact $y_i$ produced by subtask $i$. We execute $\mathcal{G}$ in topological order: at each node $v_i$, a \emph{Merger} first aggregates the upstream artifacts $\{y_j\}_{j \in \mathrm{Pa}(i)}$ into a coherent context $\mathcal{C}_i \sim p(\cdot \mid r, \mathcal{T}, \{y_j\}_{j \in \mathrm{Pa}(i)};\, \theta^{r}_i),\ r{=}\texttt{merge}$ (which reduces to $y_j$ when $|\mathrm{Pa}(i)| = 1$ and to $\emptyset$ when $\mathrm{Pa}(i) = \emptyset$); the \emph{Executor} then produces the current artifact $y_i \sim p(\cdot \mid r, \mathcal{T}, s_i, \mathcal{C}_i;\, \theta^{r}_i),\ r{=}\texttt{execute}$. Each agent call carries its own role $r$, which selects a role-conditioned LoRA composition in Sec.~\ref{subsec:mole}. We consider a general \texttt{execute} role here, while the framework readily extends to additional specialized roles (e.g., \texttt{review}, \texttt{refine}). Training proceeds sequentially to expose per-subtask credit signals, while at inference the DAG also allows independent branches to run in parallel for efficiency.

\paragraph{Rule-based verifier for reward signals.}
A task-completion-only reward is too sparse to provide sufficient supervision for fine-grained per-subtask credit assignment; we therefore apply a rule-based verifier at every DAG node to produce a step reward $u_i$ on each output artifact $y_i$, which also blocks $y_i$ from downstream propagation upon failing task-specific validity criteria (e.g., output format, basic correctness), preventing error compounding across the DAG. These step rewards $\{u_i\}_{i \in V}$ feed downstream policy optimization (Sec.~\ref{subsec:hgrpo}); per-task details are given in Appendix~\ref{app:reward-inst}.

\subsection{Prototype-Routed Mixture of Role-Subtask LoRA Experts (MoLE)}
\label{subsec:mole}
Motivated by the role-subtask architecture mismatch under parameter-efficient fine-tuning diagnosed in Sec.~\ref{subsec:prelim-mismatch}, we address it via learnable specialization on the shared frozen backbone $\theta_0$. We employ a dynamic mixture of dual-factorized LoRA experts---a role expert pool $\Phi^r$ activated per role, and a subtask expert pool $\Phi^s$ from which a subset is selected per subtask by a lightweight prototype-based semantic router---keeping training and inference cost close to a single-model system.

\paragraph{Mixture of dual-factorized LoRA experts.}
The expert pool is partitioned into role experts $\Phi^r$ (one per role $r \in \mathcal{R}$) and subtask experts $\Phi^s$ (shared across subtasks), all attached as LoRA pairs to the frozen backbone $\theta_0$ (architecture details in Appendix~\ref{app:impl-details}). For each agent call at node $v_i$ with role $r_i$, the active set $\mathcal{E}_i$ combines the role expert $e(r_i)$ with a subtask-expert subset $\mathcal{E}_{s,i}$:
\begin{equation}
\mathcal{E}_i = \{e(r_i)\} \cup \mathcal{E}_{s,i},
\qquad e(r_i) \in \Phi^r, \;\;
\mathcal{E}_{s,i} \subset \Phi^s.
\label{eq:expert-set}
\end{equation}
This active set yields the LoRA-composed projection at each adapted layer $\ell$:
\begin{equation}
W_\ell^{(i)} = W_\ell + \sum_{k \in \mathcal{E}_i} w_k \cdot \frac{\alpha}{\rho} B_{\ell, k} A_{\ell, k},
\label{eq:lora-comp}
\end{equation}
where $(A_{\ell, k}, B_{\ell, k})$ is the rank-$\rho$ LoRA pair of expert $k$, $\alpha$ is a scaling, and we use uniform aggregation weights $w_k = 1/|\mathcal{E}_i|$. The subtask-expert subset $\mathcal{E}_{s,i}$ is selected by the prototype router introduced below, except at Merger join nodes ($|\mathrm{Pa}(i)| > 1$), where it is inherited from upstream Executors as $\mathcal{E}_{s, i}^{(\texttt{merge})} = \bigcup_{j \in \mathrm{Pa}(i)} \mathcal{E}_{s, j}^{(\texttt{execute})}$, exposing the merge role to the same subtask-specific adaptation directions that produced the upstream artifacts.

\paragraph{Dynamic prototype-based semantic router.}
We introduce a router $\pi_\psi$ that scores subtask experts against per-expert learnable prototypes, inducing a soft semantic partition over the open-ended subtask space rather than a hard cluster-to-expert assignment. We encode subtask description $s_i$ via a frozen embedder, then apply a trainable projection $h_\eta$ to produce $h_\eta(s_i) \in \mathbb{R}^D$ in the same latent space as the per-expert prototypes $P \in \mathbb{R}^{K_s \times D}$ ($K_s = |\Phi^s|$):
\begin{equation}
z_k = \frac{h_\eta(s_i)^\top p_k}{\|h_\eta(s_i)\|_2\, \|p_k\|_2},
\qquad
\pi_\psi(k \mid s_i) = \mathrm{softmax}(z)_k,
\label{eq:router}
\end{equation}
where $p_k \in \mathbb{R}^D$ is the $k$-th subtask expert's prototype, $z_k$ is its cosine similarity to the projected subtask embedding, and $\pi_\psi(k \mid s_i)$ is the probability of routing $s_i$ to expert $k$; we sample top-$K$ from $\pi_\psi(\cdot \mid s_i)$ stochastically during training and greedily at inference.

\paragraph{Trainable parameters.}
The trainable components of this module are the LoRA expert matrices
$\{(A_{\ell, k}, B_{\ell, k})\}$ for all experts in $\Phi := \Phi^r \cup \Phi^s$, and the router parameters $\psi = (\eta, P)$
comprising the embedding-projection parameters $\eta$ and the prototypes $P$. Backbone $\theta_0$ remains frozen throughout training. Together $\Phi$ and $\psi$ instantiate the abstract learnable parameters $\Theta = \Phi \cup \psi$ from Sec.~\ref{sec:problem}.

\subsection{Hierarchical GRPO with Two-Layer Credit Assignment (HGRPO)}
\label{subsec:hgrpo}
Co-optimizing $\Phi$ and $\pi_\psi$ under step rewards $\{u_i\}_{i \in V}$ entangles their updates under a single shared advantage, exposing each to uncontrolled reward variation and inflating gradient variance. We address this with a \emph{two-layer credit assignment} that uses separate within-route and cross-route advantages to update experts and the router: the within-route conditional baseline strictly reduces the LoRA-expert gradient variance compared to standard GRPO (a single shared advantage for both expert and router updates), while the cross-route formulation supplies the router with a well-scaled per-route signal.

\paragraph{Two-layer credit assignment.}
At each node $v_i$ during training, we sample $B$ routes from $\pi_\psi(\cdot \mid s_i)$ (each fixing an expert combination $\mathcal{E}_i^{(b)}$), generate $M$ candidates $\{y_i^{(b, m)}\}$ per route, and obtain step rewards $\{u_i^{(b, m)}\}$ from the verifier. We form two advantages with conditional baselines:
\begin{equation}
A_{\text{within}}^{(b, m)} = \frac{u_i^{(b, m)} - \mu^{(b)}}{\sigma^{(b)} + \epsilon},
\qquad
A_{\text{cross}}^{(b)} = \frac{\mu^{(b)} - \bar\mu}{\sigma_C + \epsilon},
\label{eq:hgrpo-adv}
\end{equation}
where $\mu^{(b)}, \sigma^{(b)}$ are the within-route mean and standard deviation, $\bar\mu = \frac{1}{B}\sum_b \mu^{(b)}$ is the across-route baseline, and $\sigma_C$ is the standard deviation of $\{\mu^{(b)}\}_{b=1}^{B}$. Following GRPO~\citep{shao2024deepseekmath}, $A_{\text{within}}$ updates the LoRA experts under fixed $\mathcal{E}_i^{(b)}$ and $A_{\text{cross}}$ updates the router via the route log-likelihood:
\begin{align}
\mathcal{L}_{\text{LoRA}} &= -\!\sum_{b, m}\! A_{\text{within}}^{(b, m)} \cdot \log p\!\left(y_i^{(b, m)} \mid \mathcal{T}, r_i, s_i, \mathcal{C}_i; \mathcal{E}_i^{(b)}\right), \label{eq:lora-loss}\\
\mathcal{L}_{\text{router}} &= -\alpha_\pi \sum_b A_{\text{cross}}^{(b)} \cdot \log \pi_\psi\!\left(\mathcal{E}_{s, i}^{(b)} \mid s_i\right),
\label{eq:router-loss}
\end{align}
with $\alpha_\pi$ balancing router and LoRA gradients. The within-route baseline $\mu^{(b)}$ removes the cross-route component $\sigma_C^2$ from the LoRA-expert gradient variance; the cross-route signal $A_{\text{cross}}^{(b)}$ supplies the router with a well-scaled per-route gradient.

\paragraph{Variance reduction analysis.}
Decompose the reward variance into within- and cross-route components,
\begin{equation}
\sigma_W^2 := \mathbb{E}_b\big[\mathrm{Var}_m(u^{(b,m)} \mid b)\big] \quad\text{(within-route)},
\qquad
\sigma_C^2 := \mathrm{Var}_b\big[\mu^{(b)}\big] \quad\text{(cross-route)},
\label{eq:variance-decomp}
\end{equation}
with $\mu^{(b)} := \mathbb{E}_m[u^{(b,m)} \mid b]$ the mean reward of route $b$ (so $\mathrm{Var}(u) = \sigma_W^2 + \sigma_C^2$).

\begin{proposition}[HGRPO LoRA-expert variance reduction]
\label{prop:hgrpo-variance}
HGRPO and standard GRPO share the same expected gradient; at finite group sizes $M, B$, the HGRPO LoRA-expert gradient has strictly lower variance,
\begin{equation}
\mathrm{Var}\big[g^{\mathrm{LoRA}}_{\mathrm{HGRPO}}\big] = \mathrm{Var}\big[g^{\mathrm{LoRA}}_{\mathrm{single}}\big] - \sigma_C^2 \cdot \mathbb{E}\big[\|\nabla_\Phi \log p\|^2\big].
\label{eq:hgrpo-variance-reduction}
\end{equation}
\end{proposition}

\looseness=-1
This isolates the LoRA expert from cross-route variance it cannot influence, attributing reward signals to expert quality rather than routing decisions; the construction applies the classical conditional-baseline variance reduction~\citep{greensmith2004variance} in a two-layer form to the router--expert credit structure, and the formal proof is in Appendix~\ref{app:hgrpo-variance-proof}.
Proposition~\ref{prop:hgrpo-variance} is stated for the unnormalized
advantages, while the implemented estimator (Eq.~\ref{eq:hgrpo-adv}) further
standardizes each advantage by the per-route scale $\sigma^{(b)}+\epsilon$.
A short bridge argument shows that the variance reduction persists for this
normalized estimator in a reweighted form, up to finite-sample corrections
(Corollary~\ref{cor:hgrpo-variance-normalized},
Appendix~\ref{app:hgrpo-variance-proof}).

In conclusion, the three modules jointly instantiate the abstract framework in Sec.~\ref{sec:problem}: Sec.~\ref{subsec:tomas} supplies the decomposition $\{v_i, \mathrm{Pa}(i)\}$ and step rewards $\{u_i\}$; Sec.~\ref{subsec:mole} realises the $(r_i, s_i)$-conditional adaptation $\Delta\theta_i$ on the shared backbone $\theta_0$; and Sec.~\ref{subsec:hgrpo} optimises the learnable parameters $\Theta = \Phi \cup \psi$ with provable variance reduction on the LoRA-expert gradient.

\section{Experiment}
\label{sec:experiment}

We evaluate \textbf{MoRSE} on complex code-generation benchmarks across three
backbones, asking three questions: whether the task-oriented framework
improves inference without any training, whether MoLE together with HGRPO
delivers further gains through parameter-level specialization, and whether
the designed mechanisms behave as intended, reducing inter-agent output
redundancy, stabilizing co-optimization, and generalizing across held-out
task categories and domains.

\begin{table}[!t]
\centering
\small
\setlength{\tabcolsep}{3.5pt}
\caption{Main results on SRDD and SciCode. \textbf{MoRSE$_{\text{base}}$} is
the untrained version. \emph{All} scores inference on the full benchmark;
\emph{Test} is held out from \textbf{MoRSE}'s training and is harder on
average, so numbers are not comparable across the two sets. For each column and
backbone, the best is in \textbf{bold} (with \colorbox{morseblue}{light-blue}
on Test Set); the second-best is \underline{underlined}.}
\label{tab:main_results}
\resizebox{\textwidth}{!}{
\begin{tabular}{l|l|CCC|CCC|CCC|CCC}
\toprule
\multicolumn{1}{c|}{\multirow{3}{*}{\textbf{Type}}} & \multicolumn{1}{c|}{\multirow{3}{*}{\textbf{Method}}}
& \multicolumn{6}{c|}{\textbf{SRDD}}
& \multicolumn{6}{c}{\textbf{SciCode}} \\
\cmidrule(lr){3-8}\cmidrule(lr){9-14}
\multicolumn{1}{c|}{} & \multicolumn{1}{c|}{}
& \multicolumn{3}{c|}{\textbf{All Set}}
& \multicolumn{3}{c|}{\textbf{Test Set}}
& \multicolumn{3}{c|}{\textbf{All Set}}
& \multicolumn{3}{c}{\textbf{Test Set}} \\
\cmidrule(lr){3-5}\cmidrule(lr){6-8}\cmidrule(lr){9-11}\cmidrule(lr){12-14}
\multicolumn{1}{c|}{} & \multicolumn{1}{c|}{}
& {\scriptsize\shortstack{\textbf{Exec}\\\textbf{(\%)}}} & {\scriptsize\shortstack{\textbf{ECI}\\\textbf{Mean}}} & {\scriptsize\shortstack{\textbf{ECI}\\\textbf{Product}}}
& {\scriptsize\shortstack{\textbf{Exec}\\\textbf{(\%)}}} & {\scriptsize\shortstack{\textbf{ECI}\\\textbf{Mean}}} & {\scriptsize\shortstack{\textbf{ECI}\\\textbf{Product}}}
& {\scriptsize\shortstack{\textbf{Step}\\\textbf{Pass (\%)}}} & {\scriptsize\shortstack{\textbf{Mean Step}\\\textbf{Pass (\%)}}} & {\scriptsize\shortstack{\textbf{Problem}\\\textbf{Pass (\%)}}}
& {\scriptsize\shortstack{\textbf{Step}\\\textbf{Pass (\%)}}} & {\scriptsize\shortstack{\textbf{Mean Step}\\\textbf{Pass (\%)}}} & {\scriptsize\shortstack{\textbf{Problem}\\\textbf{Pass (\%)}}} \\
\midrule
\multicolumn{14}{c}{\cellcolor{gray!15}\textbf{Qwen3-4B-Instruct}} \\
\midrule
Single    & Base model         & \underline{67.01} & 0.686 & \underline{0.284} & 65.00 & 0.690 & \underline{0.286} & 18.05 & 17.55 & 1.25 & \underline{21.43} & 19.33 & 0.00 \\
\cmidrule{1-14}
\multirow{5}{*}{Multi} & ChatChain        & 58.33 & \underline{0.687} & 0.264 & 42.86 & 0.641 & 0.205 & 15.98 & 17.74 & 2.50 & 14.29 & 13.71 & 0.00 \\
          & MacNet             & 61.81 & 0.673 & 0.261 & 60.76 & 0.666 & 0.254 &  7.99 &  9.60 & 2.50 &  9.18 &  9.50 & 0.00 \\
          & Aflow              & 65.83 & 0.634 & 0.162 & 52.50 & 0.580 & 0.121 & \underline{19.82} & \underline{21.91} & \underline{3.75} & 18.37 & \underline{20.23} & 0.00 \\
\cmidrule(lr){2-14}
          & \textbf{MoRSE$_{\text{base}}$} & \textbf{77.00} & \textbf{0.718} & \textbf{0.299} & \underline{72.50} & \underline{0.693} & 0.268 & \textbf{25.44} & \textbf{27.22} & \textbf{6.25} & 20.41 & 19.65 & 0.00 \\
          & \textbf{MoRSE} & -     & -     & -     & \cellcolor{morseblue}\textbf{86.25} & \cellcolor{morseblue}\textbf{0.755} & \cellcolor{morseblue}\textbf{0.360} & -     & -     & -    & \cellcolor{morseblue}\textbf{29.00} & \cellcolor{morseblue}\textbf{32.50} & \cellcolor{morseblue}\textbf{10.00} \\
\midrule
\multicolumn{14}{c}{\cellcolor{gray!15}\textbf{Llama-3.1-8B-Instruct}} \\
\midrule
Single    & Base model         & 63.45 & 0.652 & \underline{0.220} & 70.51 & 0.674 & 0.251 & \underline{12.72} & 15.70 & 2.50 & 10.20 & 10.02 & 0.00 \\
\cmidrule{1-14}
\multirow{5}{*}{Multi} & ChatChain        & 66.25 & 0.628 & 0.192 & 66.25 & 0.628 & 0.192 & 12.13 & \underline{15.82} & \textbf{3.75} &  8.16 & 10.02 & 0.00 \\
          & MacNet             & 55.00 & 0.591 & 0.143 & 55.00 & 0.591 & 0.143 & 11.83 & 14.30 & 2.50 &  6.12 &  7.20 & 0.00 \\
          & Aflow              & \underline{73.37} & \underline{0.665} & 0.189 & 65.00 & 0.636 & 0.167 & 12.43 & 15.20 & 2.50 & 11.22 & 10.38 & 0.00 \\
\cmidrule(lr){2-14}
          & \textbf{MoRSE$_{\text{base}}$} & \textbf{75.00} & \textbf{0.688} & \textbf{0.252} & \underline{78.75} & \underline{0.707} & \underline{0.281} & \textbf{17.75} & \textbf{18.07} & \underline{2.50} & \underline{18.37} & \underline{14.92} & 0.00 \\
          & \textbf{MoRSE} & -     & -     & -     & \cellcolor{morseblue}\textbf{87.34} & \cellcolor{morseblue}\textbf{0.767} & \cellcolor{morseblue}\textbf{0.389} & -     & -     & -    & \cellcolor{morseblue}\textbf{26.00} & \cellcolor{morseblue}\textbf{15.50} & \cellcolor{morseblue}\textbf{5.00} \\
\midrule
\multicolumn{14}{c}{\cellcolor{gray!15}\textbf{Gemma-4-31B-IT}} \\
\midrule
Single    & Base model         & \underline{94.00} & \textbf{0.763} & \textbf{0.338} & \underline{96.25} & \underline{0.767} & \underline{0.333} & \underline{49.11} & \underline{55.93} & \underline{28.75} & 45.92 & 53.02 & 20.00 \\
\cmidrule{1-14}
\multirow{5}{*}{Multi} & ChatChain        & 89.95 & 0.706 & 0.260 & 88.75 & 0.700 & 0.256 & 48.22 & 54.74 & 27.50 & 40.82 & 49.38 & 15.00 \\
          & MacNet             & 89.95 & 0.705 & 0.246 & 88.75 & 0.688 & 0.228 & 46.45 & 51.83 & 25.00 & \underline{46.94} & \underline{53.44} & 20.00 \\
          & Aflow              & 91.46 & 0.720 & 0.272 & 87.50 & 0.690 & 0.246 & 46.45 & 52.68 & 25.00 & 45.92 & 53.02 & 20.00 \\
\cmidrule(lr){2-14}
          & \textbf{MoRSE$_{\text{base}}$} & \textbf{95.68} & \underline{0.734} & \underline{0.316} & 93.67 & 0.728 & 0.321 & \textbf{51.78} & \textbf{56.92} & \textbf{31.25} & 45.92 & 47.25 & 15.00 \\
          & \textbf{MoRSE}     & -     & -     & -     & \cellcolor{morseblue}\textbf{98.65} & \cellcolor{morseblue}\textbf{0.782} & \cellcolor{morseblue}\textbf{0.394} & -     & -     & -     & \cellcolor{morseblue}\textbf{51.02} & \cellcolor{morseblue}\textbf{56.36} & \cellcolor{morseblue}\textbf{20.00} \\
\bottomrule
\end{tabular}
}
\end{table}

\paragraph{Datasets.}
We evaluate on two code-generation benchmarks with complex task requirements: \textbf{SRDD}~\citep{qian2024chatdev}\footnote{\url{https://github.com/OpenBMB/ChatDev/tree/chatdev1.0/SRDD}} ($1{,}200$ software-requirement descriptions across $40$ application categories) and \textbf{SciCode}~\citep{tian2024scicode}\footnote{\url{https://huggingface.co/datasets/SciCode1/SciCode}} ($80$ scientific computing problems spanning $5$ scientific domains, decomposed into $338$ subproblems with executable unit tests at the step and problem level). 

\paragraph{Baselines.}
\looseness=-1
We compare against three multi-agent baselines spanning role-based, graph-structured, and search-based designs --- \textbf{ChatChain}~\citep{qian2024chatdev}, \textbf{MacNet}~\citep{qian2024macnet}, and \textbf{AFlow}~\citep{zhang2024aflow} --- together with a single-agent reference (\textbf{Base model}) that runs each backbone alone. All methods share the same Qwen3-4B-Instruct, Llama-3.1-8B-Instruct, and Gemma-4-31B-IT backbones, prompts, and decoding configuration, with each method's native revision mechanism capped at the same maximum number of attempts.

\paragraph{Evaluation metrics.}
On SRDD we follow~\citet{qian2024chatdev} and report \textbf{Exec}~$(\uparrow)$, the fraction of generated codebases that compile and run end-to-end, and \textbf{ECI}~$(\uparrow)$, a $[0,1]$-valued composite over completeness, executability, and consistency, summarised as the per-sample mean (ECI Mean) and product (ECI Product). On SciCode we follow the official protocol~\citep{tian2024scicode} and report \textbf{Step Pass}~$(\uparrow)$ (micro step-level pass rate across all subproblems), \textbf{Mean Step Pass}~$(\uparrow)$ (per-problem step accuracy averaged across problems), and \textbf{Problem Pass}~$(\uparrow)$ (fraction of problems for which all steps pass).

Detailed experimental configurations are included in Appendix~\ref{app:exp-setup}.

\subsection{Main Results}

Table~\ref{tab:main_results} reports two comparisons. The \emph{All} set scores the inference performance of methods on the full benchmark without training; the \emph{Test} split is held out from \textbf{MoRSE}'s training for direct comparison against baselines, where \textbf{MoRSE} attains the column-best score on every reported test cell.

SciCode provides both process-level performance (Step Pass, Mean Step Pass) and end-to-end success (Problem Pass); SRDD complements this with overall quality aligning with the task requirements (ECI: completeness and task-description alignment given executability). On the All set, the untrained \textbf{MoRSE}$_{\text{base}}$ (Sec.~\ref{subsec:tomas}) achieves the best overall performance comparing with all baselines, showing that the role-subtask DAG decomposition framework adds value without any training. On the held-out Test split, the trained \textbf{MoRSE} further lifts \textbf{MoRSE}$_{\text{base}}$ on all metrics, confirming the value of specialized parameter-level adaptation. 

As a side observation, prior MAS without task-specific adaptation can lag below the single-agent reference, and the gap widens on stronger backbones whose base is near-saturating.
Also, we observe that the post-training uplift is smaller on the strongest
backbone (Gemma-4-31B), where the base model is closer to saturation and, with the
LoRA configuration fixed across backbones, the trainable fraction shrinks
with model width ($0.170\%$ on Qwen3-4B, $0.109\%$ on Llama-3.1-8B, $0.048\%$
on Gemma-4-31B), making the adaptation relatively less
expressive~\citep{zeng2024expressive}. The structure and the training remain
complementary there, as the untrained framework sometimes falls slightly
below the single agent on the held-out split while training recovers the gap
and matches or exceeds it on every metric.

\emph{To conclude, on both process and end-to-end aspects, the role-subtask DAG framework improves inference performance at the task-structure level, and specialized adaptation provides a further substantial uplift at the parameter level.}

\subsection{Deeper Analysis}
\label{subsec:deeper-analysis}

\paragraph{Ablation study.}
Table~\ref{tab:main_results} shows the standalone effect of the role-subtask DAG framework (Sec.~\ref{subsec:tomas}) via \textbf{MoRSE}$_{\text{base}}$. Table~\ref{tab:ablation} ablates the remaining two core modules: (i) the \emph{MoLE architecture} (Sec.~\ref{subsec:mole}) and (ii) \emph{HGRPO training} (Sec.~\ref{subsec:hgrpo}). 

\textbf{C vs.\ D} (fixed MoLE, swap standard GRPO $\to$ HGRPO) isolates the training: HGRPO substantially recovers and extends the gain, directly verifying the hierarchical-credit design. \textbf{B vs.\ C} (under standard GRPO) reveals the converse: plugging in MoLE without hierarchical credit is fragile, indicating that experts and router can be challenging to jointly optimize under a standard optimization method.
The training dynamics behind this fragility, where the standard-GRPO
surrogate destabilizes while HGRPO remains controlled across epochs, are
shown in Appendix~\ref{app:hgrpo-dynamics}.
\textbf{A vs.\ D} then confirms the joint effect with substantial gains across all metrics; the gain is further validated by \textbf{B vs.\ D}, where \textbf{MoRSE} surpasses a non-routing fixed-LoRA baseline trained with standard GRPO whose rank is matched to \textbf{MoRSE}'s activated parameters per call.

Together, \emph{MoLE and HGRPO are tightly coupled and both are required to deliver the full improvement.}

A detailed experimental analysis for each module is provided in Appendix~\ref{app:detailed-ablation}. 

\begin{table*}[!htb]
\centering
\small
\setlength{\tabcolsep}{3.5pt}
\caption{Ablation study on SRDD and SciCode.}
\label{tab:ablation}
\resizebox{\textwidth}{!}{
\begin{tabular}{l|l|l|CCC|CCC}
\toprule
\textbf{Variant} & \textbf{Adapter} & \textbf{Training}
& \multicolumn{3}{c|}{\textbf{SRDD}}
& \multicolumn{3}{c}{\textbf{SciCode}} \\
\cmidrule(lr){4-6}\cmidrule(lr){7-9}
 & & & {\scriptsize\shortstack{\textbf{Exec}\\\textbf{(\%)}}} & {\scriptsize\shortstack{\textbf{ECI}\\\textbf{Mean}}} & {\scriptsize\shortstack{\textbf{ECI}\\\textbf{Product}}}
 & {\scriptsize\shortstack{\textbf{Step}\\\textbf{Pass (\%)}}} & {\scriptsize\shortstack{\textbf{Mean Step}\\\textbf{Pass (\%)}}} & {\scriptsize\shortstack{\textbf{Problem}\\\textbf{Pass (\%)}}} \\
\midrule
\multicolumn{9}{c}{\cellcolor{gray!15}\textbf{Qwen3-4B-Instruct}} \\
\midrule
A.\ MoRSE w/o MoLE \& HGRPO (\textbf{MoRSE}$_{\text{base}}$)                                            & \multicolumn{1}{c}{-} & \multicolumn{1}{c}{-} & 72.50 & 0.693 & 0.268 & 20.41 & 19.65 & 0.00 \\
B.\ MoRSE w/o MoLE                & fixed LoRA         & standard GRPO      & 82.50 & 0.745 & 0.329 & 25.00 & 20.25 & 5.00 \\
C.\ MoRSE w/o HGRPO                                      & MoLE               & standard GRPO      & 73.75 & 0.696 & 0.257 & 23.00 & 29.12 & 10.00 \\
\rowcolor{gray!15}
D.\ \textbf{MoRSE}                                          & \textbf{MoLE}      & \textbf{HGRPO} & \textbf{86.25} & \textbf{0.755} & \textbf{0.360} & \textbf{29.00} & \textbf{32.50} & \textbf{10.00} \\
\midrule
\multicolumn{9}{c}{\cellcolor{gray!15}\textbf{Llama-3.1-8B-Instruct}} \\
\midrule
A.\ MoRSE w/o MoLE \& HGRPO (\textbf{MoRSE}$_{\text{base}}$)                                            & \multicolumn{1}{c}{-} & \multicolumn{1}{c}{-} & 78.75 & 0.707 & 0.281 & 18.37 & 14.92 & 0.00 \\
B.\ MoRSE w/o MoLE                & fixed LoRA         & standard GRPO      & 82.50 & 0.737 & 0.326 & 23.00 & 14.00 & 0.00 \\
C.\ MoRSE w/o HGRPO                                      & MoLE               & standard GRPO      & 72.50 & 0.698 & 0.316 & 18.00 & 14.64 & 0.00 \\
\rowcolor{gray!15}
D.\ \textbf{MoRSE}                                          & \textbf{MoLE}      & \textbf{HGRPO} & \textbf{87.34} & \textbf{0.767} & \textbf{0.389} & \textbf{26.00} & \textbf{15.50} & \textbf{5.00} \\
\bottomrule
\end{tabular}
}
\end{table*}

\ensureroom{0.17\textheight}
\paragraph{Parameter-level heterogeneity gain from MoLE.}
\begin{wrapfigure}{r}{0.3\linewidth}
\centering
\includegraphics[width=\linewidth]{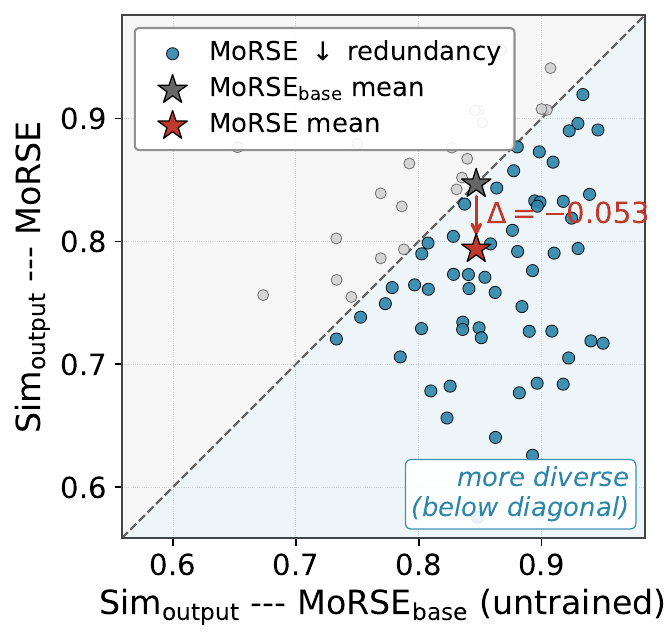}
\caption{Output $\mathrm{Sim}$.}
\label{fig:mpc_three_panel}
\vspace{-20pt}
\end{wrapfigure}
\looseness=-1
Sec.~\ref{subsec:prelim-redundancy}'s prelim diagnosed a heterogeneity deficit at the prompt level (role-prompted MAS); the role-subtask DAG framework alone mitigates it to some extent. Figure~\ref{fig:mpc_three_panel} shows that MoLE further reduces residual node-output redundancy at the \emph{parameter} level: trained dual role$\times$subtask LoRA experts shift the per-sample $\mathrm{Sim}_{\mathrm{output}}$ distribution toward more diversity, with the majority of samples landing below the diagonal, letting each node speak in its own voice rather than echoing its neighbours.

\paragraph{Out-of-distribution (OOD) generalization.}
A central premise of MoLE is that role-subtask decomposition exposes potentially \emph{generalizable} units---different tasks may share recurring subtask skills, so LoRA experts learned on training categories have the potential to compose for unseen ones.  We test this on category-level held-outs: SRDD withholds $8$ of $40$ categories; SciCode trains on Physics$+$Math$+$Material~Science and tests on Chemistry$+$Biology.

\begin{wrapfigure}{r}{0.7\linewidth}
\vspace{-12pt}
\centering
\includegraphics[width=\linewidth]{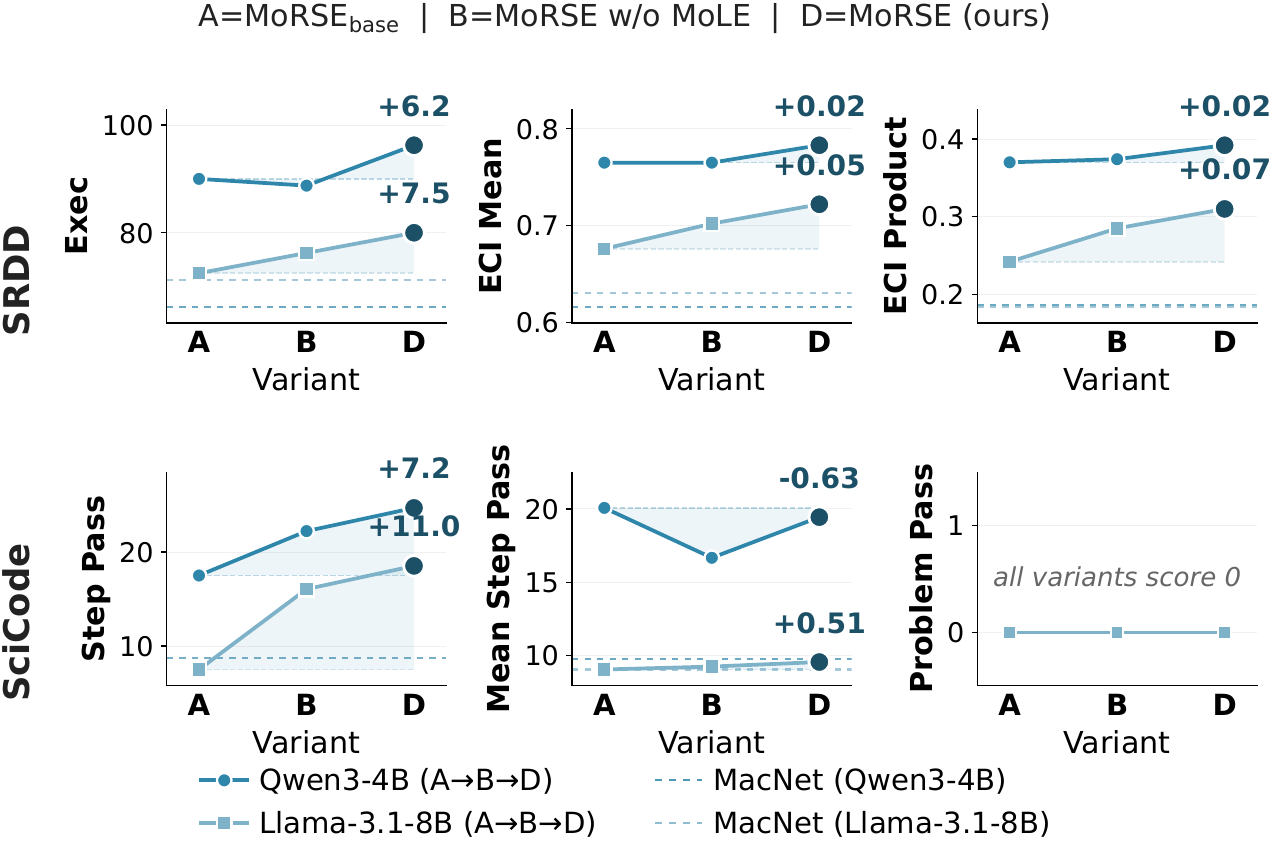}
\caption{\textbf{OOD generalization} experiments for two backbones.
Bold numbers above D are $\Delta$ vs.\ A. Per-metric tables are in
Appendix~\ref{app:ood-tables}.
}
\label{fig:ood}
\vspace{-12pt}
\end{wrapfigure}
\noindent Figure~\ref{fig:ood} contrasts four configurations on the OOD test: prior role-based DAG MAS \textsc{MacNet}, our untrained \textbf{MoRSE}$_{\text{base}}$ (A), a standard non-routing LoRA with GRPO fine-tuning baseline (B = \textbf{MoRSE} w/o MoLE), and the full trained \textbf{MoRSE} (D). Two findings emerge across both backbones and both datasets. \textbf{(i)} \emph{Training transfers to OOD}: trained \textbf{MoRSE} consistently improves over the untrained \textbf{MoRSE}$_{\text{base}}$, indicating that subtask-level skills learned on one category help held-out ones. \textbf{(ii)} \emph{Dynamic role-subtask routing matters at OOD}: \textbf{MoRSE} also surpasses the standard fine-tuning baseline, indicating that dynamic expert routing generalizes to held-out subtasks better than a fixed LoRA.

\section{Related work}

\paragraph{LLM-based multi-agent systems.}
\looseness=-1
Existing MAS coordinate cooperating agents via prompting protocols~\citep{wu2024autogen,hong2024metagpt,qian2024chatdev,du2023debate,qian2024macnet}, recently with dependency-aware DAGs~\citep{wu2025gap}; ToMAS instead exposes per-node $(\text{role}, \text{subtask})$ structure as the substrate for downstream parameter-level adaptation in MoLE.

\paragraph{Mixture of LoRA experts.}
Mixture-of-LoRA-Experts combines parameter-efficient adapters with sparse gating~\citep{dou2024loramoe}, and recent GRPO variants~\citep{shao2024deepseekmath} apply policy-gradient RL to MoE/LoRA-MoE under a single shared advantage; HGRPO instead performs a \emph{bi-level} credit decomposition with within-route and cross-route conditional baselines, sharing the same expected gradient~\citep{greensmith2004variance} but with strictly lower variance at finite group sizes (Prop.~\ref{prop:hgrpo-variance}).

\paragraph{RL for multi-agent LLM systems.}
Recent multi-agent LLM RL~\citep{moragent,stronger_mas,marti,marft} operates on static topologies and parameterizes only the role axis under a single shared advantage; a parallel line extends GRPO along the rollout-time axis~\citep{hgpo2026,gigpo2025}. MoRSE differs along three orthogonal axes---per-instance dynamic DAGs, dual role-subtask LoRA factorization on a shared backbone, and structural-axis credit decomposition via HGRPO's bi-level baselines---and is composable with rollout-time hierarchy.

An extended discussion of the above directions is provided in App.~\ref{app:related-work}.

\section{Conclusion}
We presented \textbf{MoRSE}, a task-oriented multi-agent system that addresses the heterogeneity bottleneck of role-prompted MAS via \emph{(role, subtask)-conditional} specialization at both task-structure and parameter levels. ToMAS decomposes each task into a dependency-aware DAG with per-agent (role, subtask) labels for task-structure-level specialization; MoLE attaches a dynamic mixture of role and subtask LoRA experts with a prototype-based semantic router on a shared backbone, addressing the diverse role and subtask demands that a single shared base model cannot satisfy; HGRPO stably co-optimizes experts and router via two-layer credit assignment that disentangles expert quality from routing quality. Experiments across three diverse backbones show notable whole-task and step-wise improvements, with better generalization across held-out task categories and domains.

\bibliographystyle{plainnat}
\bibliography{ref_checked}

\newpage
\appendix

\section{Notation}
\label{app:notation}

We summarize the main notation used throughout Sections~\ref{sec:preliminary}--\ref{sec:method}.

\begin{table}[!htb]
\centering
\small
\begin{tabular}{ll}
\toprule
\textbf{Symbol} & \textbf{Meaning} \\
\midrule
\multicolumn{2}{l}{\textit{Task structure and DAG (Sec.~\ref{sec:problem}, Sec.~\ref{subsec:tomas})}} \\
$\mathcal{T}$ & Complex task instance (input to the planner) \\
$\mathcal{G} = (V, E)$ & Per-instance dependency-aware DAG \\
$v_i \in V$ & Subtask node \\
$(v_i, v_j) \in E$ & Data dependency: subtask $j$ requires artifact $y_i$ \\
$\mathrm{Pa}(i)$ & Parent set of node $v_i$ \\
$r_i \in \mathcal{R}$ & Role of subtask $i$ (e.g., \texttt{execute}, \texttt{merge}) \\
$s_i$ & Natural-language subtask description \\
$y_i$ & Artifact produced at node $v_i$ \\
$\mathcal{C}_i$ & Merged context from upstream artifacts \\
\midrule
\multicolumn{2}{l}{\textit{Verifier and reward (Sec.~\ref{subsec:tomas})}} \\
$g_i \in \{0, 1\}$ & Verifier admissibility flag at node $v_i$ \\
$u_i$ & Step-level scalar reward at node $v_i$ \\
\midrule
\multicolumn{2}{l}{\textit{Backbone and experts (Sec.~\ref{sec:problem}, Sec.~\ref{subsec:mole})}} \\
$\theta_0$ & Frozen base language model parameters \\
$\Phi^r, \Phi^s$ & Pools of role and subtask LoRA experts \\
$\mathcal{E}_i \subset \Phi^r \cup \Phi^s$ & Active expert set at node $v_i$ \\
$e(r) \in \Phi^r$ & Role-to-expert mapping \\
$\theta_i$ & Agent $a_i$'s effective parameters \\
\midrule
\multicolumn{2}{l}{\textit{Router (Sec.~\ref{subsec:mole})}} \\
$\pi_\psi$ & Prototype-based semantic router policy \\
$h_\eta(s_i)$ & Subtask embedding \\
$P$ & Learnable prototype matrix \\
$\psi = (\eta, P)$ & Router parameters \\
\midrule
\multicolumn{2}{l}{\textit{HGRPO (Sec.~\ref{subsec:hgrpo})}} \\
$B$ & Number of sampled routes per node \\
$M$ & Number of candidates per route group \\
$\mathcal{E}_i^{(b)}$ & $b$-th sampled expert combination \\
$\mu^{(b)}, \sigma^{(b)}$ & Within-route mean and std of rewards \\
$\bar\mu$ & Across-route baseline \\
$A_{\text{within}}^{(b, m)}, A_{\text{cross}}^{(b)}$ & Within-route and cross-route advantages \\
$\mathcal{L}_{\text{LoRA}}, \mathcal{L}_{\text{router}}$ & LoRA and router losses \\
\bottomrule
\end{tabular}
\end{table}

\section{Related Work}
\label{app:related-work}

\subsection{LLM-based Multi-Agent System}
\looseness=-1
LLM-based multi-agent systems decompose complex tasks across cooperating agents coordinated by prompting protocols~\citep{guo2024largelanguagemodelbased}. AutoGen~\citep{wu2024autogen} provides a general conversational framework; MetaGPT~\citep{hong2024metagpt} and ChatDev~\citep{qian2024chatdev} standardise software-engineering roles; multi-agent debate~\citep{du2023debate} improves reasoning through interaction; and MacNet~\citep{qian2024macnet} scales cooperation to large agent populations. These frameworks typically rely on frozen backbones and handcrafted prompts, leaving the question of how to \emph{learn} agent-specific parameters largely open. MoRSE addresses this gap by equipping each agent with role- and subtask-conditioned LoRA~\citep{hu2021loralowrankadaptationlarge} experts that are jointly optimised together with a dynamic prototype-based~\citep{snell2017prototypicalnetworksfewshotlearning} semantic subtask router.

\paragraph{DAG-based MAS frameworks.}
A recent line of MAS work uses dependency-aware DAG decomposition to structure multi-agent collaboration: TDP~\citep{tdp2026} decouples planning from execution via per-instance sub-goal DAGs; GAP~\citep{wu2025gap} trains agent foundation models to construct sub-task graphs for parallel tool execution; and STACKPLANNER~\citep{stackplanner2026} adopts hierarchical task decomposition with task-experience memory management. Our ToMAS adopts a similar DAG abstraction but with a fundamentally different \emph{purpose}: rather than serving planning decoupling or parallel execution alone, ToMAS exposes per-node $(\text{role}, \text{subtask})$ structure as the substrate for parameter-level adaptation in MoLE, and adds requirement-coverage checking and step-level rule-based verification to make the per-node structure reliable enough for downstream parameter learning.

\subsection{Mixture of LoRA Experts}
\paragraph{Architectures.}
Mixture-of-LoRA-Experts combines parameter-efficient adapters with sparse gating~\citep{shazeer2017outrageouslylargeneuralnetworks}. LoRAMoE~\citep{dou2024loramoe} and Mixture-of-LoRAs~\citep{feng2024mixlora} show that multi-task instruction tuning benefits from per-task LoRA experts; MoLA~\citep{gao2024mola} allocates experts layer-wise; PESC~\citep{wu2024pesc} sparsifies a dense model into an MoE for instruction tuning; and hierarchical MoE~\citep{hmoe2025} introduces multi-level gating to ease routing complexity. These works are structural: training credit is still assigned via a single supervised loss or a single unified advantage.

\paragraph{Reinforcement fine-tuning for MoE / MoLE.}
A recent line applies policy-gradient RL~\citep{schulman2017proximalpolicyoptimizationalgorithms}---notably GRPO~\citep{shao2024deepseekmath}---to MoE and LoRA-MoE backbones. RO-GRPO~\citep{rogrpo2026} turns routing statistics into an auxiliary reward term to prevent expert collapse; RSPO~\citep{rspo2025} down-weights tokens with high router drift via a router-shift importance ratio; R3~\citep{r3moe2025} replays inference-time routing masks during training to close the train--inference gap; MoE-GRPO~\citep{moegrpo2026} casts expert selection as a sequential decision process in vision-language models; and PA-MoE~\citep{pamoe2026} assigns disjoint LoRA experts to task phases under a unified policy objective. \emph{All of these approaches update the router and the experts under a single shared advantage (or only shape the reward), and none decompose credit along the router--expert axis.} In contrast, our HGRPO performs a \emph{bi-level} credit decomposition: a within-route conditional baseline $(\mu_r,\sigma_r)$ for the LoRA policy and an across-route baseline $\bar\mu$ for the router. The two estimators share the same expected gradient by the conditional-baseline argument of~\citet{greensmith2004variance}, but HGRPO has strictly lower variance at finite group sizes---consistent with what we observe in the C-vs.-D ablation of Table~\ref{tab:ablation} in the main text.

\subsection{Reinforcement Learning for Multi-Agent LLM Systems}
\looseness=-1
Recent work has explored reinforcement learning and adaptation mechanisms for LLM-based multi-agent systems.
MoRAgent introduces role-conditioned LoRA adapters under single-agent training~\citep{moragent}, but does not consider multi-agent reinforcement learning or subtask-level credit assignment.
Stronger-MAS adapts group-based policy optimization to multi-agent settings via AT-GRPO, defining groups along the \emph{temporal} axis $(\text{env}, \text{agent}, \text{turn})$ with a single shared baseline~\citep{stronger_mas}; MARTI provides a unified training-inference framework that supports multiple full-parameter RL algorithms with trajectory-level group-relative advantages~\citep{marti}. Both, however, operate on \emph{static, predefined} collaboration topologies (fixed role pairs or workflow templates), parameterize agent heterogeneity at the level of full per-role LLMs, and assign credit without an explicit router--expert decomposition.
The closest exception, MARFT~\citep{marft}, formulates a Flex-MG DAG abstraction and adopts per-role LoRA, but parameterizes only the role axis, exercises its DAG abstraction in practice as a fixed agent pipeline rather than per-instance dependency graphs, and assigns credit under a single shared advantage with no router--expert decomposition. MoRSE, in contrast, executes a \emph{per-instance dynamic DAG} (whose node count, edges, and $(\text{role}, \text{subtask})$ labels are produced by the ToMAS planner per task) in topological order during training, supplying subtask-level credit signals to HGRPO.
In contrast, MoRSE (i) trains on \emph{per-instance dependency-aware DAGs} with open-ended subtask labels, (ii) factorizes parameters along \emph{both role and subtask} axes via prototype-routed LoRA experts on a shared frozen backbone, and (iii) decomposes RL credit along the \emph{router--expert structural axis} via HGRPO's bi-level conditional baselines---an axis orthogonal to AT-GRPO's $(\text{agent}, \text{turn})$ grouping and thus composable with it.

\paragraph{Hierarchical and step-level credit in GRPO.}
A parallel line extends GRPO with finer credit along the \emph{temporal} axis of agent rollouts: HGPO~\citep{hgpo2026} introduces context-aware hierarchical grouping over rollout steps for long-horizon agents; GiGPO~\citep{gigpo2025} groups anchor states across trajectories to estimate per-step advantages; GRPO-$\lambda$~\citep{grpolambda2025} refines credit at the token level; and execution-grounded schemes~\citep{execgrpo2026} exploit code-execution feedback as a step-level signal. Despite the naming overlap with HGRPO, these methods introduce hierarchy along the rollout-time axis, whereas HGRPO operates on the orthogonal \emph{structural} axis of router versus experts within a single decoding step, and can in principle be combined with any of them.

\section{Limitations and Future Work}
\label{app:limitations}

\paragraph{Limitations.}
\looseness=-1
Empirically, our evaluation covers code-generation benchmarks; broader task
families and substantially larger backbones remain to be tested.
Architecturally, the default subtask-expert pool size is intentionally
conservative, and we currently instantiate only two general roles,
\texttt{execute} and \texttt{merge}, without finer-grained roles such as
\texttt{review}, \texttt{reasoning}, or \texttt{refine}. Two further
limitations concern the evaluation protocol. First, the SRDD training reward
and the ECI evaluation metric draw on shared component sources, although
their weights and implementations differ (Appendix~\ref{app:exp-setup});
an execution-verified independent evaluation corroborates the reported
gains (Appendix~\ref{app:detailed-ablation}). Second, most reported results,
such as those in Table~\ref{tab:main_results}, are computed under a single
seed.

\paragraph{Future work.}
Scaling the subtask-expert pool---especially when paired with adaptive
per-instance allocation---should further improve coverage of heavy-tailed
subtask distributions and strengthen generalization across diverse tasks;
the pool is also a natural carrier for modality-specific experts in
multi-modal settings, where each subtask can route to the modality expert(s)
matching the input it needs to process. The framework is likewise compatible
with richer role taxonomies---each new role corresponds to an additional
role expert in $\Phi^r$ under the same MoLE composition. Finally,
parameter-level heterogeneity is orthogonal to, and composable with,
prompt-, tool-, and context-boundary-level differentiation: differentiated
tools and context accesses present exactly the divergent contexts that
demand differentiated parameters, and the expert pool and router can
condition on each agent's tool or action-space context in the same way they
condition on $(\text{role}, \text{subtask})$. Therefore, extending \textbf{MoRSE}
to such tool-use and broader agentic settings is a natural next step.

\section{Broader Impacts}
\label{app:broader-impacts}

MoRSE targets complex tasks that decompose into multiple interdependent subtasks, a setting that recurs across software engineering, scientific computing, data analysis, document and knowledge work, and decision support. By improving the capability and sample efficiency of LLM-based multi-agent systems on such tasks, the framework can lower the cost of deploying autonomous assistants in domains where progress currently depends on careful human orchestration of subtask pipelines. The work introduces no new datasets or pre-trained foundation models and does not pose specific risks beyond those already inherent to the open base models it builds on.

\section{Implementation Details}
\label{app:impl-details}

We instantiate \textbf{MoRSE} on three backbones spanning different families and sizes---Qwen3-4B-Instruct, Llama-3.1-8B-Instruct, and Gemma-4-31B-IT---as the shared backbone $\theta_0$. The planner uses the same backbone with a structured-output prompt to produce the per-instance DAG.

\paragraph{LoRA experts.} Each expert injects rank-$\rho$ LoRA pairs at
the query, value, and output projections of self-attention in the last
$L_{\text{adapt}} = 8$ transformer blocks, with $\rho = 8$ and scaling
$\alpha = 16$. The role expert pool $\Phi^r$ contains
$|\mathcal{R}| = 2$ experts (one each for \texttt{execute} and
\texttt{merge}); the subtask expert pool $\Phi^s$ contains $K_s = 4$
experts.

\paragraph{Router.} The subtask embedder mean-pools the frozen input
token embeddings of $s_i$ (truncation length $64$), applies a trainable
projection to dimension $D = 128$, and $\ell_2$-normalizes. The router
selects $K = 2$ subtask experts per agent call (top-$K$ stochastic
sampling during training, greedy at inference).

\paragraph{HGRPO training.} Each node samples $B = 4$ expert routes from $\pi_\psi$, with $M = 4$ candidates per route. Advantages are normalized by the within-route standard deviation and clipped to $\pm 5$; updates are skipped when the within-group reward std falls below $5{\times}10^{-4}$. The router gradient is scaled by $\alpha_\pi = 0.15$ relative to the LoRA gradient. Learning rates are $8{\times}10^{-5}$ (LoRA experts) and $3{\times}10^{-5}$ (router), optimized with AdamW. Sampling uses top-$p = 0.95$.

\paragraph{HGRPO instantiation across roles.}
The two-layer credit assignment in Sec.~\ref{subsec:hgrpo} applies in full to Executor calls, where the prototype router selects subtask experts and contributes a routing log-likelihood. For Merger calls, expert combinations are inherited from upstream Executors (Sec.~\ref{subsec:mole}) without any router decision, so the cross-route advantage drops out and only the within-route advantage $A_{\text{within}}^{(b,m)}$ updates the merge-role expert through $\log p(\mathcal{C}_i^{(b,m)} \mid r, \mathcal{T}, \{y_j\}_{j\in\mathrm{Pa}(i)};\, \theta^{r}_i),\ r{=}\texttt{merge}$.

\paragraph{Reward instantiations.}
\label{app:reward-inst}
At each DAG node, the rule-based verifier returns a binary admissibility flag $g_i \in \{0, 1\}$ together with a scalar score $u_i$ aggregating task-appropriate metrics, with $u_i = 0$ whenever $g_i = 0$. Inadmissible artifacts ($g_i = 0$) trigger localised regeneration with diagnostic feedback before downstream propagation, preventing error compounding across the DAG. We instantiate the scorer $u_i$ as follows.

\textit{SRDD-style code tasks.} Let $S$ denote the materialized repository artifact. We apply a hard gate that checks basic validity (e.g., presence of an entry point and a short smoke test). Conditioned on passing the gate, we compute a scalar reward as a weighted sum of task-relevant components,
\begin{equation}
R_{\text{code}}(\mathcal{T}, S) = w_{\text{exec}}\, r_{\text{exec}}(S) + w_{\text{comp}}\, r_{\text{comp}}(S) + w_{\text{cons}}\, r_{\text{cons}}(\mathcal{T}, S),
\end{equation}
where $r_{\text{exec}}$ measures executability, $r_{\text{comp}}$ measures completeness, and $r_{\text{cons}}$ measures semantic consistency between the task description and the produced codebase.

\textit{SciCode-style scientific computing tasks.} Each subtask is paired with a set of executable unit tests $\mathcal{U}_i$. The admissibility flag $g_i$ requires the generated code to parse and expose the specified function signature; conditioned on passing the gate, the scalar reward is the per-subtask test pass rate
\begin{equation}
R_{\text{sci}}(y_i, \mathcal{U}_i) = \frac{|\{t \in \mathcal{U}_i : y_i \text{ passes } t\}|}{|\mathcal{U}_i|},
\end{equation}
which aligns directly with the official Step Pass metric.

\paragraph{Compute resources.} SRDD and SciCode training on Qwen3-4B and Llama-3.1-8B uses 4$\times$H200 (141\,GB HBM each) and 4$\times$H100 (80\,GB HBM each) respectively; Gemma-4-31B training uses 4$\times$B200 (192\,GB HBM each). Each epoch takes approximately 6--12 hours; small-backbone runs use 5--10 epochs, while Gemma-4-31B uses 1--2 epochs. Per-run on-disk footprint (LoRA checkpoints, rollouts, evaluation artefacts) is under 200\,GB.

\section{Training Algorithm for MoRSE}
\label{app:training-alg}

Algorithm~\ref{alg:morse-train} summarizes one HGRPO update for a task instance $\mathcal{T}$. ToMAS first decomposes $\mathcal{T}$ into a DAG; for each executor node, the router samples $B$ expert routes and the policy generates $M$ candidates per route, after which the rule-based verifier yields step rewards $u^{(b,m)}$. HGRPO then forms hierarchical advantages (Eq.~\eqref{eq:hgrpo-adv}): within-route advantages drive the LoRA-expert gradient and cross-route advantages drive the router gradient, as derived in Sec.~\ref{subsec:hgrpo}. Merger nodes follow the same loop but only update the merge role expert with no subtask routing (see Sec.~\ref{subsec:tomas} and Appendix~\ref{app:reward-inst}).

\begin{algorithm}[!t]
\caption{MoRSE training step (per task instance $\mathcal{T}$).}
\label{alg:morse-train}
\begin{algorithmic}[1]
\REQUIRE Frozen backbone $\theta_0$; LoRA experts $\Phi = \Phi^r \cup \Phi^s$; prototype router $\pi_\psi$; rule-based verifier $\mathcal{V}$; group sizes $(B, M)$; PPO clip $\epsilon_{\text{clip}}$; step size $\eta$.
\STATE Decompose via ToMAS: $G = (V, E),\ \{(r_i, s_i, \mathrm{Pa}(i))\}_{i \in V} \leftarrow \mathrm{ToMAS}(\mathcal{T})$ \hfill\COMMENT{Sec.~\ref{subsec:tomas}}
\STATE Initialize gradient accumulators $g^{\Phi} \leftarrow 0$,\ $g^{\psi} \leftarrow 0$
\FOR{each node $v_i \in V$ in topological order}
    \STATE Build context $\mathcal{C}_i \leftarrow \{y_j\}_{j \in \mathrm{Pa}(i)}$
    \FOR{route $b = 1, \dots, B$}
        \STATE Sample expert subset $\mathcal{E}_{s,i}^{(b)} \sim \pi_\psi(\cdot \mid s_i)$ \hfill\COMMENT{top-$K$ stochastic during training; greedy at inference}
        \STATE $\theta_i^{(b)} \leftarrow \theta_0 + \Delta\Phi\big(\{e(r_i)\} \cup \mathcal{E}_{s,i}^{(b)}\big)$
        \FOR{candidate $m = 1, \dots, M$}
            \STATE Sample $y^{(b,m)} \sim p(\cdot \mid \mathcal{T}, s_i, \mathcal{C}_i;\, \theta_i^{(b)})$
            \STATE Compute step reward $u^{(b,m)} \leftarrow \mathcal{V}(y^{(b,m)}, s_i)$
        \ENDFOR
        \STATE Within-route stats: $\mu^{(b)} \leftarrow \tfrac{1}{M}\sum_m u^{(b,m)}$,\ $\sigma^{(b)} \leftarrow \mathrm{std}_m\big(u^{(b,m)}\big)$
    \ENDFOR
    \STATE Cross-route stats: $\bar\mu \leftarrow \tfrac{1}{B}\sum_b \mu^{(b)}$,\ $\sigma_C \leftarrow \mathrm{std}_b\big(\mu^{(b)}\big)$
    \STATE $A_{\text{within}}^{(b,m)} \leftarrow \big(u^{(b,m)} - \mu^{(b)}\big) / (\sigma^{(b)} + \epsilon)$ \hfill\COMMENT{LoRA-expert advantage}
    \STATE $A_{\text{cross}}^{(b)} \leftarrow \big(\mu^{(b)} - \bar\mu\big) / (\sigma_C + \epsilon)$ \hfill\COMMENT{router advantage}
    \STATE $g^{\Phi} \mathrel{+}= \nabla_\Phi\, J^{\text{PPO-clip}}_{\epsilon_{\text{clip}}}\!\big(\{A_{\text{within}}^{(b,m)}\};\, \pi_{\theta_i^{(b)}}\big)$
    \STATE $g^{\psi} \mathrel{+}= \nabla_\psi\, J^{\text{PPO-clip}}_{\epsilon_{\text{clip}}}\!\big(\{A_{\text{cross}}^{(b)}\};\, \pi_\psi(\cdot\mid s_i)\big)$
    \STATE Select one $y_i \in \{y^{(b,m)}\}$ for downstream context propagation \hfill\COMMENT{e.g., highest-reward route}
\ENDFOR
\STATE $\Phi \leftarrow \Phi + \eta\, g^{\Phi}$;\quad $\psi \leftarrow \psi + \eta\, g^{\psi}$
\ENSURE Updated $(\Phi, \psi)$.
\end{algorithmic}
\end{algorithm}

\section{HGRPO Variance Reduction: Proposition and Proof}
\label{app:hgrpo-variance-proof}

We formalise the gradient-variance reduction of the two-layer credit assignment in Sec.~\ref{subsec:hgrpo} relative to standard GRPO, which applies the empirical mean $\bar u = \tfrac{1}{BM}\sum_{b,m} u_i^{(b,m)}$ as the only baseline for both expert and router updates.

\paragraph{Setup.}
At a fixed node $v_i$, treat the route index $b \in \{1, \ldots, B\}$ and the candidate index $m \in \{1, \ldots, M\}$ as random draws (routes from $\pi_\psi(\cdot \mid s_i)$, candidates from the policy under $\mathcal{E}_i^{(b)}$). Let $u^{(b,m)} := u_i^{(b,m)}$ for brevity. Define the conditional baselines $\mu^{(b)} := \mathbb{E}_m[u \mid b]$ and $\bar\mu := \mathbb{E}_b[\mu^{(b)}]$, with the law-of-total-variance decomposition
\[
\sigma_W^2 := \mathbb{E}_b[\mathrm{Var}_m(u \mid b)],
\qquad
\sigma_C^2 := \mathrm{Var}_b[\mu^{(b)}],
\qquad
\mathrm{Var}(u) = \sigma_W^2 + \sigma_C^2.
\]

\paragraph{Proposition (informal).}
The HGRPO and standard GRPO gradient estimators in Eqs.~\eqref{eq:lora-loss}--\eqref{eq:router-loss} share the same expected gradient. The variance reduction, however, differs between the two updates.

\textbf{LoRA-expert update (real reduction).} At finite group sizes $M, B$, the HGRPO LoRA-expert gradient has strictly lower per-term variance than standard GRPO, with reduction
\[
\mathrm{Var}\big[g^{\mathrm{LoRA}}_{\mathrm{single}}\big] - \mathrm{Var}\big[g^{\mathrm{LoRA}}_{\mathrm{HGRPO}}\big]
= \sigma_C^2 \cdot \mathbb{E}\big[\|\nabla_\Phi \log p\|^2\big] \;\ge\; 0.
\]
Because each candidate produces its own gradient $\nabla_\Phi \log p^{(b,m)}$, the per-term reduction $\sigma_C^2$ accumulates additively across all $BM$ candidates.

\paragraph{Proof sketch (LoRA-expert reduction).}
Unbiasedness follows from the standard score-function identity: for any baseline that is independent of the action conditional on the corresponding context, $\mathbb{E}[(R - \text{baseline}) \nabla \log p] = \mathbb{E}[R \nabla \log p]$ since $\mathbb{E}[\nabla \log p] = 0$. Both $\mu^{(b)}$ (constant in $m$ given $b$) and $\bar\mu$ (constant across actions) satisfy this requirement.

For the variance reduction, write
\[
g^{\mathrm{LoRA}}_{\mathrm{HGRPO}} = (u - \mu^{(b)}) \nabla_\Phi \log p,
\qquad
g^{\mathrm{LoRA}}_{\mathrm{single}} = (u - \bar\mu) \nabla_\Phi \log p,
\]
and observe $u - \bar\mu = (u - \mu^{(b)}) + (\mu^{(b)} - \bar\mu)$, where the two terms are uncorrelated (within-route vs.\ cross-route deviation, orthogonal under conditional expectation). Taking variances and using independence of the score function from the cross-route deviation $\mu^{(b)} - \bar\mu$ within the conditional-baseline regime~\citep{greensmith2004variance},
\[
\mathrm{Var}[g^{\mathrm{LoRA}}_{\mathrm{single}}] = \mathrm{Var}[g^{\mathrm{LoRA}}_{\mathrm{HGRPO}}] + \sigma_C^2 \cdot \mathbb{E}[\|\nabla_\Phi \log p\|^2].
\]

\paragraph{Remarks.}
(i) With the $\sigma^{(b)}$ normalisation included, the LoRA-expert variance reduction carries the reweighted factor made precise in Corollary~\ref{cor:hgrpo-variance-normalized} below. (ii) The reduction is largest when the cross-route reward variance $\sigma_C^2$ dominates---i.e., when the routed expert combinations differ substantially in quality, which is precisely the regime where credit assignment matters most.

\begin{corollary}[normalized form]
\label{cor:hgrpo-variance-normalized}
Under the conditions of Proposition~\ref{prop:hgrpo-variance}, with per-route
standardization, replacing the global baseline with the within-route baseline
reduces the per-term LoRA-expert gradient variance by
\begin{equation}
\mathbb{E}_b\!\left[\frac{(\mu^{(b)}-\bar\mu)^2}{(\sigma^{(b)}+\epsilon)^2}\right]\cdot\mathbb{E}\big[\|\nabla_\Phi \log p\|^2\big]\;\ge\;0.
\label{eq:hgrpo-variance-normalized}
\end{equation}
\end{corollary}

\paragraph{Bridge to the normalized estimator (proof of Corollary~\ref{cor:hgrpo-variance-normalized}).}
With per-route standardization, the two estimators become
$g^{\mathrm{LoRA}}_{\mathrm{HGRPO}} = \frac{u-\mu^{(b)}}{\sigma^{(b)}+\epsilon}\,\nabla_\Phi\log p$
and
$g^{\mathrm{LoRA}}_{\mathrm{single}} = \frac{u-\bar\mu}{\sigma^{(b)}+\epsilon}\,\nabla_\Phi\log p$,
both divided by the same route-constant scale $\sigma^{(b)}+\epsilon$.
Writing
$\frac{u-\bar\mu}{\sigma^{(b)}+\epsilon} = \frac{u-\mu^{(b)}}{\sigma^{(b)}+\epsilon} + \frac{\mu^{(b)}-\bar\mu}{\sigma^{(b)}+\epsilon}$,
the two terms remain uncorrelated under the same conditional-independence
conditions as in the unnormalized proof, and the shared within-route
component cancels in the variance difference. The reduction is therefore
carried by the second term alone,
\[
\mathrm{Var}\big[g^{\mathrm{LoRA}}_{\mathrm{single}}\big] -
\mathrm{Var}\big[g^{\mathrm{LoRA}}_{\mathrm{HGRPO}}\big]
= \mathbb{E}_b\!\left[\frac{(\mu^{(b)}-\bar\mu)^2}{(\sigma^{(b)}+\epsilon)^2}\right]\cdot
\mathbb{E}\big[\|\nabla_\Phi \log p\|^2\big] \;\ge\; 0,
\]
which is exactly Corollary~\ref{cor:hgrpo-variance-normalized}.

\looseness=-1
Two qualifications apply. Replacing population means with empirical group
statistics preserves both conclusions up to $O(1/M)$ finite-sample
corrections; equal expectation holds up to two benign scaling
effects, a $(1-1/M)$ factor absorbed by the learning rate and the empirical
per-route scale acting as an adaptive step size, both shared by
group-relative estimators including standard GRPO. Advantage clipping and
the low-variance-group skip are engineering stabilizers outside the
analysis.

\paragraph{Router update (scaling-equivalent at the batch level).}
The router score function $\nabla_\psi \log \pi_\psi(\mathcal{E}^{(b)})$ is constant in $m$, so the flat-GRPO batch-summed router gradient collapses to $M\sum_b(\hat\mu^{(b)} - \bar u)\nabla_\psi \log \pi^{(b)}$. Combined with $\bar u = \bar\mu$, this equals $M$ times the unnormalised HGRPO router gradient $\sum_b(\hat\mu^{(b)} - \bar\mu)\nabla_\psi \log \pi^{(b)}$; the two estimators therefore have identical signal-to-noise ratio at the batch level. The substantive contribution of HGRPO's per-route formulation is the $\sigma_C$ normalisation, providing a well-scaled router gradient that does not depend on the group size $M$.

\paragraph{Scope of the guarantees.}
In summary, the formal guarantees concern the LoRA-expert update: HGRPO
shares the expected gradient of standard GRPO and strictly reduces the
LoRA-expert gradient variance at finite group sizes
(Proposition~\ref{prop:hgrpo-variance}), a reduction that persists for the
implemented normalized estimator in the reweighted form of
Corollary~\ref{cor:hgrpo-variance-normalized}, up to $O(1/M)$ finite-sample
corrections. For the router, the analysis establishes a weaker property:
the batch-level router gradient is scaling-equivalent to that of flat GRPO,
and the per-route $\sigma_C$ normalization supplies a well-scaled signal
whose magnitude does not depend on the group size $M$; no variance-reduction
claim is made for the router update. The stability of the joint
co-optimization of experts and router is an empirical finding rather than a
theorem: it is supported by the training dynamics in
Appendix~\ref{app:hgrpo-dynamics}, where the HGRPO surrogate remains within
$\pm 0.05$ across all epochs while the standard-GRPO variant destabilizes,
and by the \textbf{C vs.\ D} ablation in Table~\ref{tab:ablation}.

\section{HGRPO Training Dynamics}
\label{app:hgrpo-dynamics}

Figure~\ref{fig:hgrpo_training} compares the full method (\emph{ours}) to an ablation that removes hierarchical credit (\emph{$-$hgrpo}, standard GRPO over merged trajectories) on Qwen3-4B SRDD training. \textit{Left:} per-epoch GRPO surrogate (updated steps only, symlog scale). The surrogate measures the importance-ratio--advantage product and is expected to stay near zero under stable on-policy updates; deviations indicate either ratio drift or advantage blow-up. \emph{Ours} remains within $\pm 0.05$ across all $10$ epochs, indicating that the within-route conditional baseline keeps both quantities controlled. The standard-GRPO variant crashes to $\approx -11.5$ at epochs $7$--$8$ before partially recovering, consistent with the variance-reduction analysis in Sec.~\ref{subsec:hgrpo} and Appendix~\ref{app:hgrpo-variance-proof}: without the conditional baseline, advantages computed across merged trajectories accumulate cross-route variance that destabilises the policy ratio. \textit{Right:} the downstream consequence on SRDD evaluation step-pass rate---\emph{ours} rises monotonically from $\sim\!21$\% at epoch~$1$ to $\sim\!28$\% at epoch~$10$, while the \emph{$-$hgrpo} curve plateaus earlier and is eventually overtaken. Together these two views confirm that hierarchical credit does not merely enable MoLE to fit the data: it is required for stable optimization under the DAG/merge rollout structure.

\begin{figure}[!htb]
\centering
\includegraphics[width=0.48\linewidth]{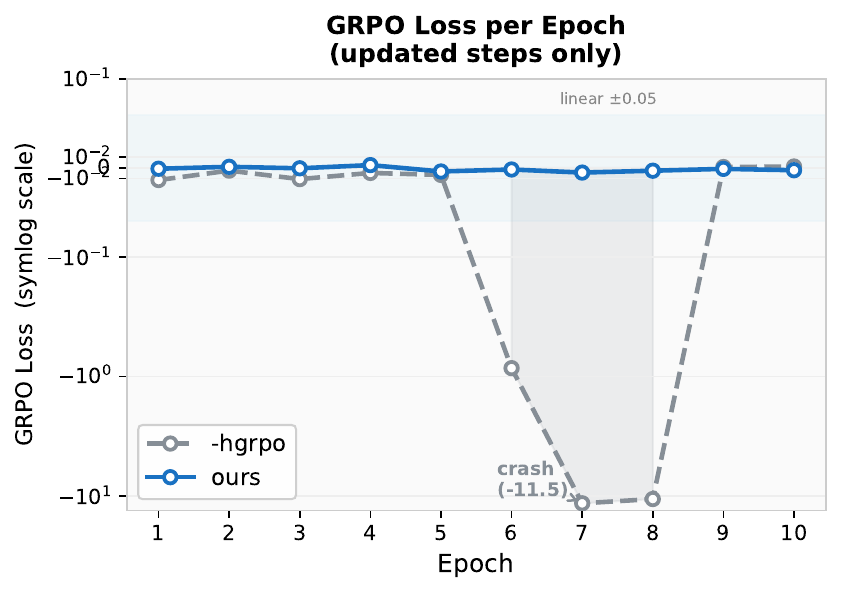}\hfill
\includegraphics[width=0.48\linewidth]{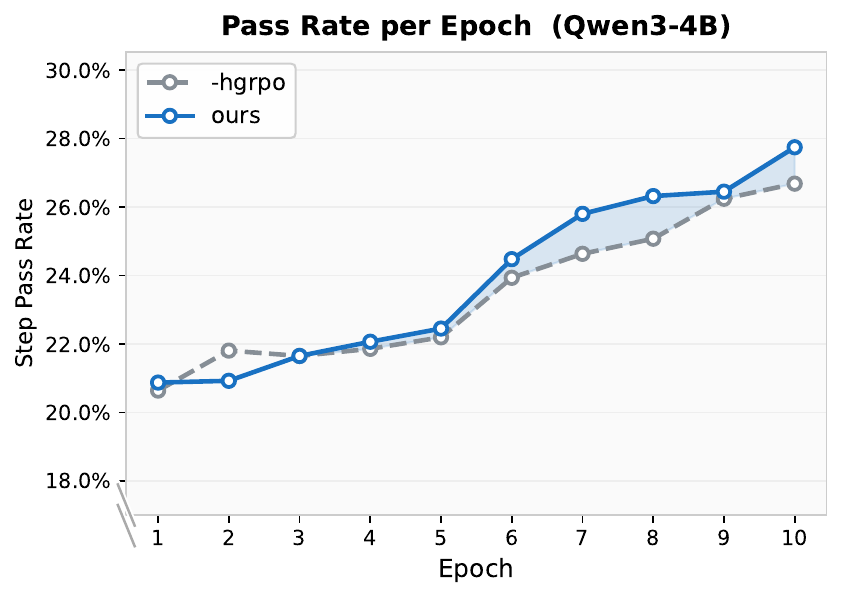}
\caption{\textbf{HGRPO training dynamics on Qwen3-4B (SRDD).} \textit{Left:} per-epoch GRPO loss (updated steps only, symlog scale). \textit{Right:} step-pass rate per epoch.}
\label{fig:hgrpo_training}
\end{figure}

\section{Experimental Setup Details}
\label{app:exp-setup}

This appendix expands the experimental setup of Sec.~\ref{sec:experiment} (Datasets, Baselines, and Evaluation Metrics).

\paragraph{Dataset splits.}
For SRDD we sample from the $1{,}200$ examples with a $3{:}2$ train/test ratio for in-distribution evaluation. The out-of-distribution split is category-disjoint: $8$ of the $40$ categories (\texttt{Strategy\_Game}, \texttt{Science}, \texttt{Health\_Fitness}, \texttt{Graphics}, etc.) are withheld entirely from training, yielding a $4{:}1$ train/test sample ratio. For SciCode we use the official $60/20$ problem split for IID evaluation, and a $64/16$ problem-disjoint OOD split that withholds the \texttt{Chemistry} and \texttt{Biology} domains while training only on \texttt{Physics}, \texttt{Math}, and \texttt{Material Science}.

\paragraph{Baselines.}
\textbf{ChatChain}~\citep{qian2024chatdev} is a sequential CEO/CTO/Programmer/Counselor pipeline that produces code through structured chat across two phases (language selection followed by coding). \textbf{MacNet}~\citep{qian2024macnet} organizes agents on a $5$-node DAG with per-node review and an aggregation step that merges the outputs; both ChatChain and MacNet run on the official ChatDev codebase implementation~\citep{qian2024chatdev}. \textbf{AFlow}~\citep{zhang2024aflow} is a search-based agentic-workflow generator that constructs the orchestration graph via Monte-Carlo Tree Search over operator templates. \textbf{MoRSE}$_{\text{base}}$ is the untrained MoRSE framework (row~A of Table~\ref{tab:ablation}) and serves as an anchor that isolates the contribution of training on the same DAG. All baselines and our method share Qwen3-4B-Instruct, Llama-3.1-8B-Instruct, and Gemma-4-31B-IT backbones served via HuggingFace Transformers in bfloat16, with the same prompts and decoding configuration ($8192$-token budget, temperature $0.2$, top-$p$ $0.95$). Each baseline runs with its own native revision mechanism: code-review-and-test revision in ChatChain, verification-informed review-and-rewrite along edges in MacNet, test-and-revise operators in AFlow, and the same verify-and-regeneration attempts for the single-agent reference. The maximum number of attempts is fixed to $3$ across all pipelines on SRDD; on SciCode, the Table~\ref{tab:main_results} evaluation includes no revision for any method, while the verifier still provides step-level rewards during training.

\paragraph{Evaluation metrics.}
On SRDD we report two scoring metrics from~\citet{qian2024chatdev}: \textbf{Exec}~$(\uparrow)$, the fraction of generated codebases that compile and run end-to-end; and \textbf{ECI}~$(\uparrow)$, a $[0,1]$-valued composite over per-sample completeness $r_{\text{comp}}$, executability $r_{\text{exec}}$, and consistency $r_{\text{cons}}$, summarised across the suite as
\begin{equation}
\mathrm{ECI}_{\text{mean}} = \tfrac{1}{3}\big(r_{\text{exec}} + r_{\text{comp}} + r_{\text{cons}}\big), \qquad \mathrm{ECI}_{\text{product}} = r_{\text{exec}} \cdot r_{\text{comp}} \cdot r_{\text{cons}}.
\end{equation}
\looseness=-1
On SciCode we report the official~\citep{tian2024scicode} \textbf{Step Pass}~$(\uparrow)$ (micro step-level pass rate), \textbf{Mean Step Pass}~$(\uparrow)$ (per-problem step accuracy averaged across problems), and \textbf{Problem Pass}~$(\uparrow)$ (fraction of problems for which all steps pass), identical to the main-text protocol of Sec.~\ref{sec:experiment}.

\paragraph{Reward--metric relationship (SRDD).}
\looseness=-1
We fix $w_{\text{exec}} = 0.5$, $w_{\text{comp}} = 0.5$, and
$w_{\text{cons}} = 1.0$ (split as $0.7$ for consistency with the overall
task description and $0.3$ for the node's own subtask); the weights are
principled defaults fixed a priori from the design intent and kept unchanged
across all trained variants and backbones. The training reward further
differs from the ECI evaluation metrics in both implementation and
composites. The consistency term in the reward is a purely lexical
bag-of-words cosine between the (sub)task description and the stripped code,
without any learned model, whereas the ECI evaluation computes an embedding
cosine with \texttt{gte-Qwen2-7B-instruct}, so the training signal never
observes the evaluation scorer. The composites also differ:
$\mathrm{ECI}_{\text{mean}}$ averages the three components with equal
weights while the reward is non-uniformly weighted, and
$\mathrm{ECI}_{\text{product}}$ multiplies the three components, an
aggregation form that never appears in the reward.
An execution-verified check
(Appendix~\ref{app:detailed-ablation}) verifies that the gains are not an
artifact of this alignment.

\section{Detailed Experimental Analysis}
\label{app:detailed-ablation}

This appendix expands the ablation study of Sec.~\ref{subsec:deeper-analysis} with a module-level progression, budget-accounted trained comparisons, a router-utilization analysis, and an execution-verified independent evaluation; the training-dynamics comparison behind the stability ablation is in Appendix~\ref{app:hgrpo-dynamics}.

\paragraph{Module-level progression.}
Table~\ref{tab:progression} traces the progression from the single agent to the untrained framework and the trained \textbf{MoRSE} on the held-out Test split, where the framework-level uplift covers the task decomposition together with the verifier-guided regeneration, and the training-level uplift covers the proposed parameter specialization and credit assignment. The training-level uplift exceeds the framework-level uplift on every metric on both benchmarks under the same revision budget; within it, the part beyond ordinary RL fine-tuning (D vs.\ the budget-matched single-LoRA controls in Table~\ref{tab:param-budget}) remains positive on every metric.

\begin{table}[!htb]
\centering
\small
\caption{Progression from the single agent to \textbf{MoRSE}$_{\text{base}}$ and the trained \textbf{MoRSE} on the held-out Test split (Qwen3-4B; Table~\ref{tab:main_results} numbers). Uplifts are relative (\%), except SciCode Problem Pass in percentage points.}
\label{tab:progression}
\resizebox{\textwidth}{!}{
\begin{tabular}{l|l|c|c}
\toprule
\# & \textbf{Setting} & \textbf{SRDD (Exec / ECI-M / ECI-P)} & \textbf{SciCode (Step / Mean Step / Problem Pass)} \\
\midrule
1 & Single-agent & 65.00 / 0.690 / 0.286 & 21.43 / 19.33 / 0.00 \\
2 & MoRSE$_{\text{base}}$ & 72.50 / 0.693 / 0.268 & 20.41 / 19.65 / 0.00 \\
3 & MoRSE & \textbf{86.25 / 0.755 / 0.360} & \textbf{29.00 / 32.50 / 10.00} \\
\midrule
-- & Framework-level uplift ($1 \to 2$) & $+11.5\%$ / $+0.4\%$ / $-6.3\%$ & $-4.8\%$ / $+1.7\%$ / $+0.0$ pp \\
-- & Training-level uplift ($2 \to 3$) & $\mathbf{+19.0\%}$ / $\mathbf{+8.9\%}$ / $\mathbf{+34.3\%}$ & $\mathbf{+42.1\%}$ / $\mathbf{+65.4\%}$ / $\mathbf{+10.0}$ \textbf{pp} \\
\bottomrule
\end{tabular}
}
\end{table}

\paragraph{All trained comparisons with budget accounting.}
Table~\ref{tab:param-budget} collects all trained comparisons together with the full budget accounting. Each expert is a rank-$8$ LoRA on the q/v/o projections of the last $8$ blocks, and each call activates $1$ role $+$ top-$2$ subtask experts, giving the activated vs.\ total accounting ($3.24$M / $6.49$M). \textit{Trainable parameters:} B exactly matches \textbf{MoRSE}'s per-call activation, and B$+$ matches the total budget. \textit{Rollout budget:} identical total per call, sampled as one flat credit group under standard GRPO and as $4$ routes $\times$ $4$ candidates per route under HGRPO; for C, each of the $16$ candidates still samples its own expert combination, while the credit group stays flat. \textit{Optimization budget:} identical training data, reward, and schedule, under the same AdamW optimizer with the same LoRA learning rate ($8 \times 10^{-5}$) and top-$p = 0.95$ sampling (Appendix~\ref{app:impl-details}). D outperforms both single-LoRA controls (B, B$+$), and doubling the single LoRA's rank to the total-parameter match brings no consistent further gain under the same data and training settings, indicating that the gains are not simply an effect of a larger parameter budget under either the activated or the total accounting.

\begin{table}[!htb]
\centering
\small
\caption{All trained comparisons on the SRDD held-out Test split (Qwen3-4B). Rollouts are per call; Router = the dynamic role--subtask prototype router; GRPO = standard GRPO.}
\label{tab:param-budget}
\resizebox{\textwidth}{!}{
\begin{tabular}{l|l|c|c|c|c|c|c|c}
\toprule
\# & \textbf{Setting} & \textbf{Rank} & \textbf{Act.\ params} & \textbf{Total params} & \textbf{Rollouts} & \textbf{Router} & \textbf{Credit} & \textbf{Exec / ECI-M / ECI-P} \\
\midrule
1 & Single-agent (Table~\ref{tab:main_results}) & -- & -- & -- & -- & -- & -- & 65.00 / 0.690 / 0.286 \\
2 & Single-agent LoRA RL & $24 \times 1$ & \textbf{3.24M} & 3.24M & \textbf{16} & none & GRPO & 70.75 / 0.722 / 0.321 \\
3 & A. MoRSE$_{\text{base}}$ (Table~\ref{tab:ablation}) & -- & -- & -- & -- & -- & -- & 72.50 / 0.693 / 0.268 \\
4 & B. MoRSE w/o MoLE (act.\ params matched; Table~\ref{tab:ablation}) & $24 \times 1$ & \textbf{3.24M} & 3.24M & \textbf{16} & none & GRPO & 82.50 / 0.745 / 0.329 \\
5 & B$+$. MoRSE w/o MoLE (total params matched) & $48 \times 1$ & 6.49M & \textbf{6.49M} & \textbf{16} & none & GRPO & 81.25 / 0.741 / 0.348 \\
6 & C. MoRSE w/o HGRPO ($2$-role $+$ $4$-subtask MoLE; Table~\ref{tab:ablation}) & $8 \times (2{+}4)$ & \textbf{3.24M} & \textbf{6.49M} & \textbf{16} & \checkmark & GRPO & 73.75 / 0.696 / 0.257 \\
7 & D. MoRSE ($2$-role $+$ $4$-subtask MoLE; Table~\ref{tab:ablation}) & $8 \times (2{+}4)$ & \textbf{3.24M} & \textbf{6.49M} & \textbf{16} ($4 \times 4$) & \checkmark & HGRPO & \textbf{86.25 / 0.755 / 0.360} \\
\bottomrule
\end{tabular}
}
\end{table}

\looseness=-1
\textit{Isolation map.} \textbf{HGRPO:} row 6 vs.\ row 7 (fixed MoLE, standard GRPO $\to$ HGRPO) isolates the hierarchical credit assignment. \textbf{Dynamic role--subtask routing:} row 4 vs.\ row 6 isolates the routing architecture under standard GRPO, fragile without hierarchical credit (Sec.~\ref{subsec:deeper-analysis}), while rows 4--5 vs.\ row 7 give the net contribution of dynamic routing combined with HGRPO under matched budgets. \textbf{The multi-agent structure itself:} row 1 vs.\ row 3 without training, and row 2 vs.\ row 4 under identical training, where the two rows share the same data, trainable-parameter, rollout, and optimization budgets and differ only in the multi-agent structure.

\paragraph{Router utilization and the expert-count diagnostic.}
After training, the router does not collapse to a nearly fixed route: every
expert enters the greedy top-$2$ for some subtasks, the soft routing entropy
averages $1.374$ nats and the realized top-$2$ utilization entropy is
$1.009$ nats, both against the maximum $\ln 4 = 1.386$ and the latter well
above the $\ln 2 = 0.693$ of a fixed pair, and the most common pair carries
$76.1\%$ of the subtasks while $23.9\%$ route elsewhere
(Table~\ref{tab:utilization}). The pair pattern is a learned
shared-plus-specialist structure rather than a degeneration, one generalist
expert staying active while the second slot switches with the subtask,
mirroring the deliberate shared-expert design of
DeepSeekMoE~\citep{dai2024deepseekmoe}. Training further includes two
anti-collapse mechanisms, an entropy regularizer on each routing
distribution (weight $0.02$) and an orthogonality penalty on the expert
prototypes (weight $0.05$), so every expert receives gradient signal
throughout training.

\begin{table}[!htb]
\centering
\small
\caption{Expert- and pair-level utilization of the trained model (Table~\ref{tab:ablation} row D, Qwen3-4B, SRDD test).}
\label{tab:utilization}
\begin{tabular}{c|c|c||c|c}
\toprule
\textbf{Expert} & \textbf{Selected in top-2} & \textbf{Mean routing prob.} & \textbf{Expert pair (top-2)} & \textbf{Share of subtasks} \\
\midrule
0 & 21.3\% & 0.227 & (1, 3) & 76.1\% \\
1 & 76.1\% & 0.238 & (0, 3) & 21.3\% \\
2 & 2.6\% & 0.219 & (2, 3) & 2.6\% \\
3 & 100.0\% & 0.316 & (0,1), (0,2), (1,2) & 0\% \\
\bottomrule
\end{tabular}
\end{table}

\looseness=-1
The cluster count in Sec.~\ref{subsec:prelim-mismatch} is diagnostic, while
the expert count is a design choice: $K{=}8$ is the
clustering resolution used to demonstrate that heterogeneous subtask demands
exist, not a design target, and with top-$2$ routing the four subtask
experts act as a compositional basis, providing $\binom{4}{2}=6$ subtask
combinations crossed with the two role experts. The choice of four was fixed in advance, never tuned on validation or test performance, and we do not claim it is optimal.

\paragraph{Execution-verified independent evaluation (SRDD).}
To verify that the SRDD gains are not an artifact of the residual alignment
between the training reward and the ECI metrics
(Appendix~\ref{app:exp-setup}), we build an execution-verified evaluation
that is not encoded in the training reward and apply it to variants A, B,
and D of Table~\ref{tab:ablation} on the SRDD test split (Qwen3-4B). For
each task, $3$ to $5$ acceptance checks are derived from the task
description alone, blind to all systems; every generated program is then
executed in a scripted interactive session, and its runtime transcript is
scored against the checks by a deterministic pattern match and by a stricter
LLM pass/fail grading. Both the checks and the LLM grading come from a model
family disjoint from the trained backbone and the evaluation embedder, and
the training reward never sees the checks, the runtime protocol, or either
scoring rule. Table~\ref{tab:exec-verified} reports both metrics;
\textbf{MoRSE} improves over both variants under either scoring rule.

\begin{table}[!htb]
\centering
\small
\caption{Execution-verified evaluation of the Table~\ref{tab:ablation} variants on the SRDD test split (Qwen3-4B). Coverage is the fraction of acceptance checks satisfied.}
\label{tab:exec-verified}
\begin{tabular}{l|c|c}
\toprule
\textbf{Variant} & \textbf{Coverage (pattern-matched)} & \textbf{Coverage (LLM-scored, Claude Opus 4.8)} \\
\midrule
A. MoRSE$_{\text{base}}$ & 0.501 & 0.167 \\
B. MoRSE w/o MoLE & 0.582 & 0.202 \\
D. MoRSE & \textbf{0.619} & \textbf{0.224} \\
\bottomrule
\end{tabular}
\end{table}

\section{OOD Generalization Tables}
\label{app:ood-tables}

This appendix provides per-row OOD tables corresponding to Fig.~\ref{fig:ood} in the main body. $\Delta^{\mathrm{gain}}_{\mathrm{OOD}}=\text{variant}-\text{A}$ computed on the same OOD samples isolates the training contribution from intrinsic domain difficulty.

\begin{table*}[!htb]
\centering
\small
\setlength{\tabcolsep}{3pt}
\caption{OOD generalization on \textbf{SRDD} (8 categories held out from training).}
\label{tab:ood_srdd}
\resizebox{\textwidth}{!}{
\begin{tabular}{l|l|CCC|CCC}
\toprule
& & \multicolumn{3}{c|}{\textbf{OOD test}} & \multicolumn{3}{c}{\textbf{$\Delta^{\mathrm{gain}}_{\mathrm{OOD}}$ = variant $-$ A}} \\
\textbf{Variant} & \textbf{Description}
& {\scriptsize\shortstack{\textbf{Exec}\\\textbf{(\%)}}} & {\scriptsize\shortstack{\textbf{ECI}\\\textbf{Mean}}} & {\scriptsize\shortstack{\textbf{ECI}\\\textbf{Product}}}
& {\scriptsize\shortstack{\textbf{$\Delta$}\\\textbf{Exec}}} & {\scriptsize\shortstack{\textbf{$\Delta$ ECI}\\\textbf{Mean}}} & {\scriptsize\shortstack{\textbf{$\Delta$ ECI}\\\textbf{Product}}} \\
\midrule
\multicolumn{8}{c}{\cellcolor{gray!15}\textbf{Qwen3-4B-Instruct}} \\
\midrule
0.\ MacNet                               & \textit{prior MAS baseline}         & 66.25          & 0.616          & 0.186          & $-$23.75 & $-$0.149 & $-$0.184 \\
A.\ MoRSE w/o MoLE \& HGRPO (\textbf{MoRSE}$_{\text{base}}$)              & \textit{untrained reference}        & 90.00          & 0.765          & 0.370          & 0.00 & 0.000 & 0.000 \\
B.\ MoRSE w/o MoLE                       & \textit{standard-training baseline} & 88.75          & 0.765          & 0.374          & $-$1.25 & 0.000 & +0.004 \\
\rowcolor{gray!15}
D.\ \textbf{MoRSE}                       & \textit{ours}                       & \textbf{96.25} & \textbf{0.783} & \textbf{0.392} & \textbf{+6.25} & \textbf{+0.018} & \textbf{+0.022} \\
\midrule
\multicolumn{8}{c}{\cellcolor{gray!15}\textbf{Llama-3.1-8B-Instruct}} \\
\midrule
0.\ MacNet                               & \textit{prior MAS baseline}         & 71.25          & 0.630          & 0.184          & $-$1.25 & $-$0.046 & $-$0.058 \\
A.\ MoRSE w/o MoLE \& HGRPO (\textbf{MoRSE}$_{\text{base}}$)              & \textit{untrained reference}        & 72.50          & 0.676          & 0.242          & 0.00 & 0.000 & 0.000 \\
B.\ MoRSE w/o MoLE                       & \textit{standard-training baseline} & 76.25          & 0.702          & 0.285          & +3.75 & +0.026 & +0.043 \\
\rowcolor{gray!15}
D.\ \textbf{MoRSE}                       & \textit{ours}                       & \textbf{80.00} & \textbf{0.722} & \textbf{0.310} & \textbf{+7.50} & \textbf{+0.046} & \textbf{+0.068} \\
\bottomrule
\end{tabular}
}
\end{table*}

\begin{table*}[!htb]
\centering
\small
\setlength{\tabcolsep}{3pt}
\caption{OOD generalization on \textbf{SciCode} (Chemistry$+$Biology withheld; training uses only Physics$+$Math$+$Material~Science).}
\label{tab:ood_scicode}
\resizebox{\textwidth}{!}{
\begin{tabular}{l|l|CCC|CCC}
\toprule
& & \multicolumn{3}{c|}{\textbf{OOD test}} & \multicolumn{3}{c}{\textbf{$\Delta^{\mathrm{gain}}_{\mathrm{OOD}}$ = variant $-$ A}} \\
\textbf{Variant} & \textbf{Description}
& {\scriptsize\shortstack{\textbf{Step}\\\textbf{Pass (\%)}}} & {\scriptsize\shortstack{\textbf{Mean Step}\\\textbf{Pass (\%)}}} & {\scriptsize\shortstack{\textbf{Problem}\\\textbf{Pass (\%)}}}
& {\scriptsize\shortstack{\textbf{$\Delta$ Step}\\\textbf{Pass}}} & {\scriptsize\shortstack{\textbf{$\Delta$ Mean}\\\textbf{Step Pass}}} & {\scriptsize\shortstack{\textbf{$\Delta$ Problem}\\\textbf{Pass}}} \\
\midrule
\multicolumn{8}{c}{\cellcolor{gray!15}\textbf{Qwen3-4B-Instruct}} \\
\midrule
0.\ MacNet                               & \textit{prior MAS baseline}         & 8.75      & 9.79      & 0.00    & $-$8.75      & $-$10.28      & 0.00    \\
A.\ MoRSE w/o MoLE \& HGRPO (\textbf{MoRSE}$_{\text{base}}$)              & \textit{untrained reference}        & 17.50 & 20.07 & 0.00 & 0.00   & 0.00   & 0.00 \\
B.\ MoRSE w/o MoLE                       & \textit{standard-training baseline} & 22.22 & 16.67 & 0.00 & +4.72  & $-$3.40 & 0.00 \\
\rowcolor{gray!15}
D.\ \textbf{MoRSE}                       & \textit{ours}                       & \textbf{24.69} & \textbf{19.44} & \textbf{0.00} & \textbf{+7.19} & \textbf{$-$0.63} & \textbf{0.00} \\
\midrule
\multicolumn{8}{c}{\cellcolor{gray!15}\textbf{Llama-3.1-8B-Instruct}} \\
\midrule
0.\ MacNet                               & \textit{prior MAS baseline}         & 8.75     & 9.06     & 0.00    & +1.25      & 0.00      & 0.00    \\
A.\ MoRSE w/o MoLE \& HGRPO (\textbf{MoRSE}$_{\text{base}}$)              & \textit{untrained reference}        & 7.50  & 9.06  & 0.00 & 0.00   & 0.00   & 0.00 \\
B.\ MoRSE w/o MoLE                       & \textit{standard-training baseline} & 16.05 & 9.26  & 0.00 & +8.55  & +0.20  & 0.00 \\
\rowcolor{gray!15}
D.\ \textbf{MoRSE}                       & \textit{ours}                       & \textbf{18.52} & \textbf{9.57}  & \textbf{0.00} & \textbf{+11.02} & \textbf{+0.51}  & \textbf{0.00} \\
\bottomrule
\end{tabular}
}
\end{table*}

\newpage

\section{Role-Subtask Architecture Mismatch in the Base LM: Procedure, Diagnostics, and Coverage Limit}
\label{app:role-subtask-mismatch}

\looseness=-1
This appendix supports the empirical motivation in Sec.~\ref{subsec:prelim-mismatch}, where we claim that different roles and subtasks impose mismatched parameter-adaptation demands on the base LM. We operationalise this mismatch via rank-$\rho$ subspace divergence between $(\text{role}, \text{subtask})$ groups.
\S\ref{app:role-subtask-subspace-procedure} states the diagnostic procedure;
\S\ref{app:subspace-diag} provides supplementary evidence (per-layer stability, cross-benchmark robustness);
\S\ref{app:lora-subspace-conflict} proves a coverage-limit theorem in the same divergence quantity, formalising why a single shared rank-$\rho$ LoRA cannot cover both groups simultaneously.

\subsection{Diagnostic Procedure}
\label{app:role-subtask-subspace-procedure}

Our diagnostic is inspired by prior representation-similarity analyses that summarise neural activations by their low-rank principal subspaces and compare two such subspaces via Frobenius/principal-angle distances on the Grassmann manifold~\citep{raghu2017svcca,kornblith2019similarity,hamm2008grassmann}. Whereas those works compare representations \emph{across networks or layers}, we apply the same machinery \emph{within a single base LM at a fixed layer}, conditioned on $(\text{role}, \text{subtask})$ group membership, in order to test whether different groups carry their adaptation signal in different rank-$\rho$ directions. The role diagnostic (Figure~\ref{fig:prelim-role-subtask-mismatch}(a)) and the subtask diagnostic (Figure~\ref{fig:prelim-role-subtask-mismatch}(b,c)) share the same SVD-subspace + shuffle-baseline pipeline; they differ only in how prompts are constructed and how groups are defined.

\paragraph{Backbone and target layers.}
We use Qwen3-4B-Instruct in \texttt{float16}. The diagnostic targets the last $N{=}8$ decoder layers, matching the LoRA injection range used by MoLE on $\{\texttt{q\_proj}, \texttt{v\_proj}, \texttt{o\_proj}\}$. For each prompt we run a single no-grad forward pass and capture the input activation to each target layer at the last token position---exactly what the LoRA module sees---giving a tensor $A^{(L)} \in \mathbb{R}^{N_p \times d}$ per layer.

\paragraph{Subspace and divergence measure.}
For each layer $L$ and group $g$, restrict $A^{(L)}$ to rows in $g$, mean-center, and take the top-$\rho$ right singular vectors $V_g^{(L)} \in \mathbb{R}^{d \times \rho}$ with $\rho{=}8$ matching the LoRA rank. For every group pair $(g_a, g_b)$,
\[
\mathrm{div}^{(L)}(g_a, g_b) \;=\; 1 - \tfrac{1}{\rho}\big\|\,(V_{g_a}^{(L)})^{\top} V_{g_b}^{(L)}\big\|_F^2 \;\in\; [0, 1],
\]
where $0$ corresponds to identical rank-$\rho$ subspaces and $1$ to mutually orthogonal ones.

\paragraph{Shuffle baseline.}
For each layer we randomly permute group labels (preserving group sizes), recompute the group SVD subspaces, and average the off-diagonal divergence; repeating $T{=}20$ trials yields a baseline mean$\pm\sigma$. The real-vs-shuffle gap quantifies how much divergence is genuinely group-conditional rather than an artifact of any size-preserving partition.

\paragraph{Role diagnostic ($K{=}2$).}
We collect $200$ \texttt{execute}-node subtask descriptions from SRDD inference traces. Each is rendered into \emph{two} prompts that differ only in a role marker (\texttt{ROLE=execute} vs.\ \texttt{ROLE=merge}) and a short role-specific instruction tail; subtask content is held identical, so the divergence isolates role-conditioning from subtask content.

\paragraph{Subtask diagnostic ($K{=}8$).}
We take all SRDD \texttt{execute}-node prompts, build TF-IDF features, and run KMeans with $K{=}8$ clusters (a finer granularity than the $K_s{=}4$ subtask experts; the clustering is a diagnostic instrument rather than a prescription for the expert count). Each prompt uses a fixed \texttt{ROLE=execute} marker so role is held constant and cluster identity is the only varying factor; up to $50$ prompts per cluster are sampled. Off-diagonal entries of the $K{\times}K$ divergence matrix appear in Figure~\ref{fig:prelim-role-subtask-mismatch}(c) (representative layer) and Figure~\ref{fig:app-srdd-per-layer} (all 8 layers).

\subsection{Supplementary Diagnostics}
\label{app:subspace-diag}

\paragraph{All eight LoRA layers.}
Figure~\ref{fig:app-srdd-per-layer} shows the $K{\times}K$ subtask-cluster divergence heatmap for every LoRA-injected layer. The pattern of high- and low-divergence cluster pairs is visually stable across all eight layers, confirming that the layer instance shown in Figure~\ref{fig:prelim-role-subtask-mismatch}(c) is not cherry-picked.

\begin{figure}[!t]
  \centering
  \includegraphics[width=0.82\linewidth]{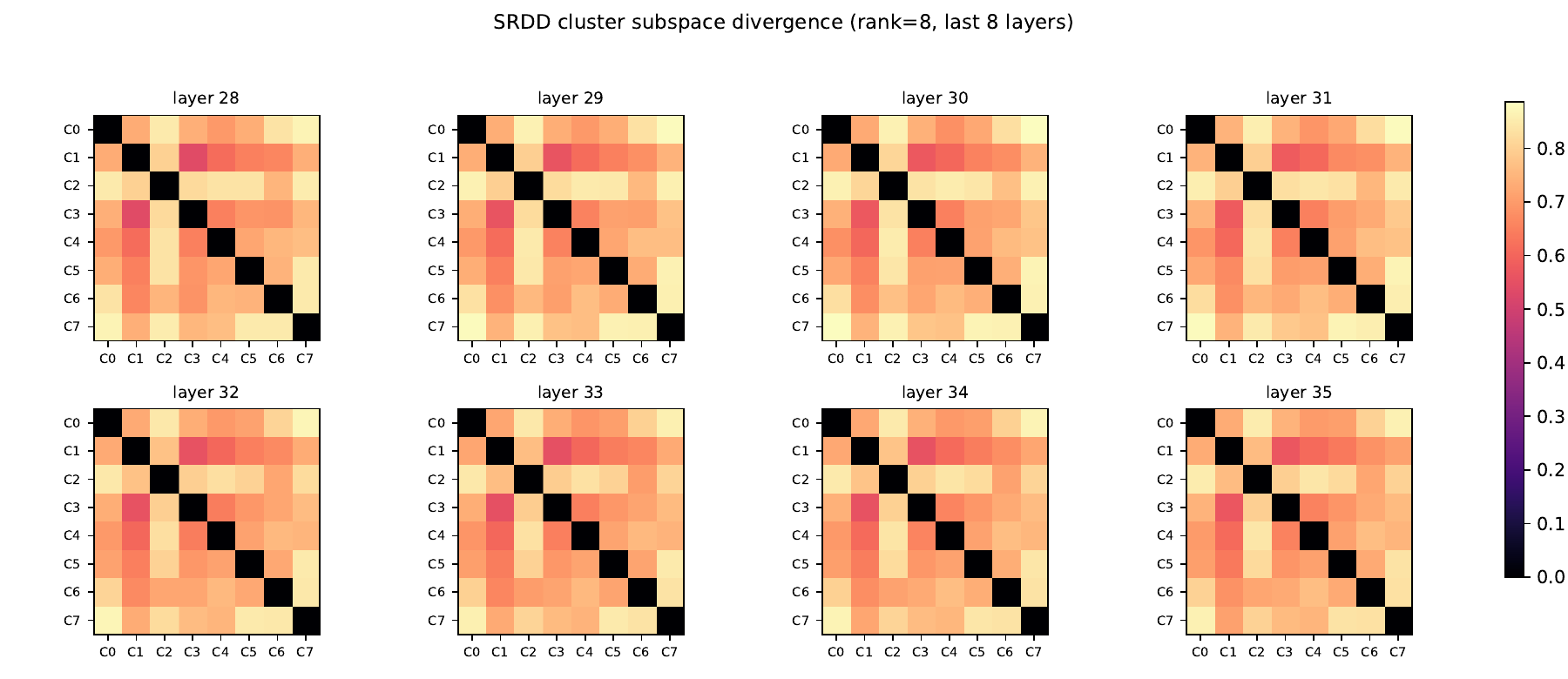}
  \caption{\textbf{SRDD subspace divergence across all 8 LoRA layers.} Each sub-panel is a cluster-pair heatmap ($1-\|V_a^\top V_b\|_F^2/\rho$) at one layer; off-diagonal structure is nearly constant across panels.}
  \label{fig:app-srdd-per-layer}
\end{figure}

\paragraph{Robustness across benchmarks (SciCode).}
Repeating the diagnostic on SciCode reproduces the SRDD finding: the real-vs-shuffle gap is $0.29$ ($0.73$ vs.\ $0.44$), $\pm\sigma$-bands are non-overlapping at every layer, and the layer-wise cluster-pair heatmap shows the same off-diagonal heterogeneity (Figure~\ref{fig:app-scicode}).

\begin{figure}[!t]
  \centering
  \begin{subfigure}[b]{0.42\linewidth}
    \includegraphics[width=\linewidth]{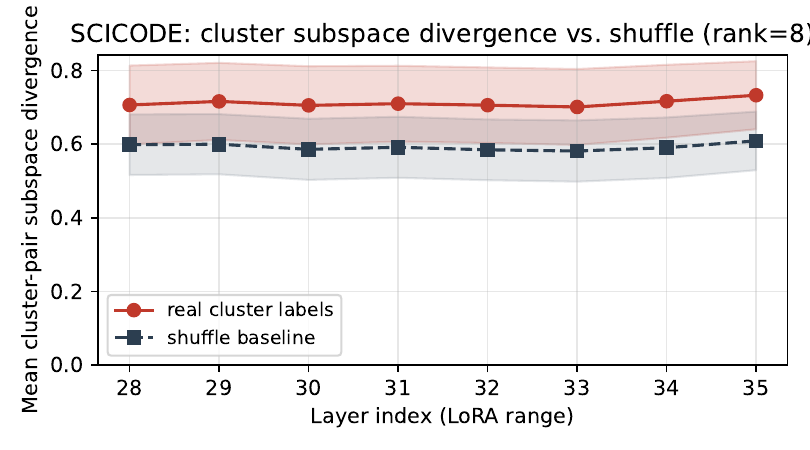}
    \caption{SciCode: real vs.\ shuffle.}
    \label{fig:app-scicode-a}
  \end{subfigure}
  \hfill
  \begin{subfigure}[b]{0.42\linewidth}
    \includegraphics[width=\linewidth]{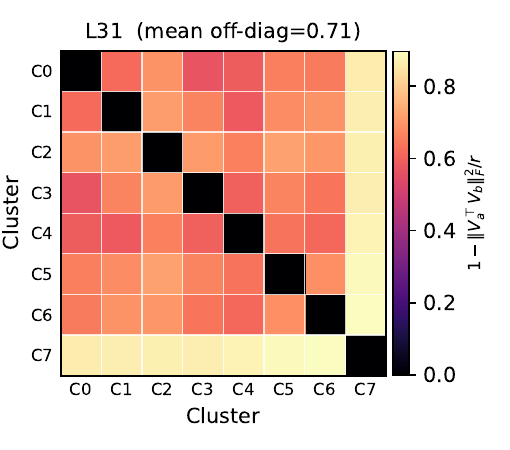}
    \caption{SciCode: representative layer.}
    \label{fig:app-scicode-b}
  \end{subfigure}
  \caption{\textbf{SciCode robustness.} Repeating the SRDD diagnostic on SciCode: real-vs-shuffle gap $0.29$ with non-overlapping $\pm\sigma$-bands across layers (left), and the same off-diagonal cluster-pair structure at the representative layer (right).}
  \label{fig:app-scicode}
\end{figure}

\subsection{LoRA Subspace Coverage Limit}
\label{app:lora-subspace-conflict}

The empirical observation in Sec.~\ref{subsec:prelim-mismatch} is that role- and subtask-conditioned updates of the base LM live in distinct rank-$\rho$ subspaces. We now show that no single shared rank-$\rho$ LoRA adapter can simultaneously cover both subspaces. The assumption is stated in the same divergence quantity as Figure~\ref{fig:prelim-role-subtask-mismatch}, so the empirical and theoretical statements speak about the same object.

\begin{assumption}[Group subspace divergence]
\label{assump:lora-noncollinear}
At some layer of the base LM, there exist two groups (e.g., two roles, or two subtask clusters) whose group-optimal rank-$\rho$ LoRA update subspaces $\mathcal{V}_a, \mathcal{V}_b \subseteq \mathbb{R}^{d}$, with orthonormal bases $V_a, V_b \in \mathbb{R}^{d \times \rho}$, satisfy
\[
\mathcal{D}(V_a, V_b) \;\triangleq\; 1 - \tfrac{1}{\rho}\|V_a^{\top} V_b\|_F^2 \;\ge\; \Delta \;>\; 0.
\]
\end{assumption}

\noindent This is exactly the divergence in Figure~\ref{fig:prelim-role-subtask-mismatch}(a,b); $\Delta$ is the empirically measured real-label value at each layer.

\begin{theorem}[LoRA Subspace Coverage Limit]
\label{thm:lora-subspace-conflict}
Let $\theta_1, \dots, \theta_\rho \in [0, \pi/2]$ be the principal angles between $\mathcal{V}_a$ and $\mathcal{V}_b$. For any shared rank-$\rho$ LoRA update subspace with orthonormal basis $V_s \in \mathbb{R}^{d \times \rho}$,
\[
\bar p(V_s) \;\triangleq\; \tfrac{\|V_s^\top V_a\|_F^2 + \|V_s^\top V_b\|_F^2}{2\rho}
\;\le\; \tfrac{1}{2} + \tfrac{1}{2\rho}\textstyle\sum_{i=1}^\rho \cos\theta_i
\;\le\; \tfrac{1+\sqrt{1-\Delta}}{2} \;<\; 1.
\]
The first inequality is tight when $V_s$ stacks the top-$\rho$ eigenvectors of $V_a V_a^\top + V_b V_b^\top$, i.e., the principal-angle bisecting subspace $\mathrm{span}\{(u_i+w_i)/\|u_i+w_i\|\}_{i=1}^\rho$ of $(\mathcal{V}_a, \mathcal{V}_b)$. The second is tight iff all principal angles are equal. As an immediate corollary, $\min\bigl(\|V_s^\top V_a\|_F^2,\,\|V_s^\top V_b\|_F^2\bigr)/\rho \le (1+\sqrt{1-\Delta})/2$, so at least one of the two groups is strictly under-covered by any single shared rank-$\rho$ adapter.
\end{theorem}

\begin{proof}
\textbf{(1) Trace formulation.}
Using $\|V_s^\top V\|_F^2 = \mathrm{tr}(V_s^\top V V^\top V_s)$ and trace cyclicity,
\[
\|V_s^\top V_a\|_F^2 + \|V_s^\top V_b\|_F^2
= \mathrm{tr}\!\big(V_s^\top M V_s\big), \qquad
M \;\triangleq\; V_a V_a^\top + V_b V_b^\top \;\in\; \mathbb{R}^{d\times d},
\]
where $M$ is symmetric and positive semi-definite.

\textbf{(2) Ky Fan trace maximum.}
Let $\lambda_1(M) \ge \lambda_2(M) \ge \cdots \ge \lambda_d(M) \ge 0$ be the eigenvalues of $M$. By the Ky Fan trace maximum principle,
\[
\max_{V_s^\top V_s = I_\rho} \mathrm{tr}\!\big(V_s^\top M V_s\big) \;=\; \sum_{i=1}^{\rho} \lambda_i(M),
\]
attained by stacking the top-$\rho$ orthonormal eigenvectors of $M$ in $V_s$.

\textbf{(3) Spectrum of $M$ via principal angles.}
Write the SVD $V_a^\top V_b = U \,\mathrm{diag}(\cos\theta_i)\, W^\top$ with $U, W \in \mathbb{R}^{\rho\times\rho}$ orthogonal and $\theta_i \in [0,\pi/2]$. Set $\bar V_a \triangleq V_a U$ and $\bar V_b \triangleq V_b W$; their columns $u_i, w_i \in \mathbb{R}^d$ form orthonormal bases of $\mathcal{V}_a$ and $\mathcal{V}_b$ with $u_i^\top w_j = \cos\theta_i\,\delta_{ij}$ (canonical correlation pairs), and $V_a V_a^\top = \bar V_a \bar V_a^\top$, similarly for $b$. For each $i$, $M$ leaves $\mathrm{span}(u_i, w_i)$ invariant: when $\theta_i \in (0, \pi/2]$ the two vectors are linearly independent and a direct computation gives eigenvalues $1+\cos\theta_i$ (eigenvector $\propto u_i+w_i$) and $1-\cos\theta_i$ (eigenvector $\propto u_i-w_i$); when $\theta_i = 0$ we have $u_i=w_i$ and the subspace collapses to one eigenvalue $1+\cos\theta_i = 2$. On $(\mathcal{V}_a + \mathcal{V}_b)^\perp$, $M$ acts as zero. Hence the non-zero spectrum of $M$ is
\[
\{1+\cos\theta_i\}_{i=1}^\rho \;\cup\; \{1-\cos\theta_i : \theta_i > 0\}.
\]
Since $\cos\theta_i \in [0,1]$, every $1+\cos\theta_i \ge 1 \ge 1-\cos\theta_j$, so the top $\rho$ eigenvalues are precisely $\{1+\cos\theta_i\}_{i=1}^\rho$, giving
\[
\sum_{i=1}^{\rho} \lambda_i(M) \;=\; \rho \,+\, \sum_{i=1}^\rho \cos\theta_i.
\]
Combining with (1)--(2) yields the first inequality of the theorem (with equality at the principal-bisector subspace).

\textbf{(4) Bounding $\sum_i \cos\theta_i$ by $\Delta$.}
Since $\|V_a^\top V_b\|_F^2 = \sum_{i=1}^\rho \cos^2\theta_i$, Assumption~\ref{assump:lora-noncollinear} gives $\sum_i \cos^2\theta_i \le \rho(1-\Delta)$. By Cauchy--Schwarz applied to the all-ones vector and $(\cos\theta_i)_{i=1}^\rho$,
\[
\sum_{i=1}^\rho \cos\theta_i \;\le\; \sqrt{\rho \cdot \textstyle\sum_i \cos^2\theta_i} \;\le\; \rho\sqrt{1-\Delta},
\]
with equality iff all $\cos\theta_i$ are equal. Substituting yields the second inequality of the theorem; $\Delta>0$ implies $\sqrt{1-\Delta}<1$, so $\bar p(V_s) < 1$.

\textbf{(5) Min-coverage corollary.}
$\min(x_1, x_2) \le (x_1+x_2)/2$ applied to $x_g = \|V_s^\top V_g\|_F^2/\rho$ gives the corollary.
\end{proof}

\paragraph{Connection to dynamic LoRA composition.}
The assumption parameter $\Delta$ is exactly the rank-$\rho$ subspace divergence reported in Figure~\ref{fig:prelim-role-subtask-mismatch}(a,b), so the empirical measurement and the theoretical bound describe the same quantity (no bridging assumption needed). Plugging in the observed subtask-cluster divergence ${\sim}0.75$ across the deeper layers (Figure~\ref{fig:prelim-role-subtask-mismatch}(b)) gives $\bar p(V_s) \le (1+\sqrt{0.25})/2 = 0.75$: under this idealized reading, a single shared rank-$\rho$ adapter would cover at most $75\%$ of the average target subspace, and the corollary forces at least one of the two groups to be strictly under-covered. The dynamic LoRA composition in Sec.~\ref{subsec:mole} sidesteps this limit by maintaining \emph{multiple} rank-$\rho$ subspaces---factorized role experts $\{\Delta W^{\mathrm{role}}_r\}$ and subtask experts $\{\Delta W^{\mathrm{subtask}}_s\}$---and routing $(\Delta W^{\mathrm{role}}_r + \Delta W^{\mathrm{subtask}}_s)$ conditionally on $(r_i, s_i)$ via the prototype router, so each call accesses an adaptation subspace matched to its $(\text{role}, \text{subtask})$ context rather than a globally compromised one. The argument extends to $N>2$ groups by replacing $M$ with $\sum_{g=1}^N V_g V_g^\top$; the average-coverage maximum is then $\sum_{i=1}^{\rho}\lambda_i(M)/(N\rho)$, which decreases as more pairwise divergences exceed zero.

\end{document}